\documentclass[12pt]{article}
\usepackage[usenames,dvipsnames]{color}
\usepackage{a4wide}
\usepackage{amsmath}
\usepackage{amsfonts}
\usepackage{amssymb}
\usepackage{amsthm}
\usepackage{mathtools}
\usepackage{hyperref}
\usepackage{graphicx}
\usepackage{latexsym}
\usepackage[T1]{fontenc}
\usepackage{color}
\usepackage[titletoc]{appendix}
\usepackage{color}

\def\al{\alpha}

\def\id{{\rm{id}}}
\def\sk{\vskip .4cm}
\def\noi{\noindent}

\def\al{\alpha}

\def\sk{\vskip .4cm}
\def\noi{\noindent}

\def\al{\alpha}

\let \part\partial

\def\epsi{\varepsilon}

\def\de{\delta}

\def\part{\partial}

\def\sk{\vskip .4cm}

\def\noi{\noindent}

\def\X0{X^0}

\def\al{\alpha}

\def\epsi{\varepsilon}

\def\de{\delta}

\def\square{{\,\lower0.9pt\vbox{\hrule \hbox{\vrule height 0.2 cm
\hskip 0.2 cm \vrule height 0.2 cm}\hrule}\,}}

\def\Ahat{{\hat A}}
\def\epsihat{{\hat \epsi}}
\def\dehat{{\hat \de}}
\def\Fhat{{\hat F}}
\def\fhat{{\hat f}}
\def\ghat{{\hat g}}
\def\phihat{{\hat \phi}}
\def\Psihat{{\hat \Psi}}

\def\*{\star}
\def\m*{\star^{op}}

\newcommand{\eqa}{\begin{eqnarray}}
\newcommand{\en}{\end{equation}}
\newcommand{\ena}{\end{eqnarray}}
\newcommand{\enn}{\nonumber \end{equation}}

\newcommand{\eq}[1]{\begin{equation}
                     \begin{split} #1 \end{split}
                     \end{equation}}

\newcommand{\CCi}{{C}^\infty}
\newcommand{\CE}{{\cal E}}

\newcommand{\CL}{{\cal L}}

\newcommand{\SWA}{\hat A}
\newcommand{\SWE}{\hat \varepsilon}
\newcommand{\SWP}{\hat \phi}
\newcommand{\SWPA}{\hat \Psi}
\newcommand{\SWF}{\hat F}
\newcommand{\SWV}{\hat \delta}

\newcommand{\TSWp}{\phi^\theta}

\newcommand{\Hop}{\phi^\theta}
\newcommand{\Tphi}{\Hop}

\newcommand{\SWAp}{{\hat A}^{{\!\!\:\!\!\!\phantom{I}}^{{'}}}}
\newcommand{\SWEp}{{\hat \varepsilon}^{\!\!\;'}}
\newcommand{\SWPp}{{\hat \phi}^{{\!\!\:\!\!\!\phantom{I}}^{{'}}}}
\newcommand{\SWVp}{{\hat \delta}^{{\!\;\!\!\!\!\phantom{I}}^{{'}}}}

\newcommand{\AAA}{{A}}
\newcommand{\BBB}{{B}}
\numberwithin{equation}{section}
\def\TT2{T}

\newtheorem{definition}{Definition}
\newtheorem{theorem}{Theorem}
\newtheorem{rem}[theorem]{Remark}
\newtheorem{lemma}[theorem]{Lemma}
\newtheorem{prop}[theorem]{Proposition}

\newcommand{\longhookrightarrow}{\lhook\joinrel\xrightarrow{\phantom{----}}}

\newcommand{\longhookuparrow}{\mathrel{\rotatebox[origin=c]{90}{$\:\:\longhookrightarrow\:\;$}}}

\newcommand{\M}{\mathcal{M}}

\begin{document}
\thispagestyle{empty}
\vspace*{-1.5cm}

\begin{flushright}
  {}
\end{flushright}

\vspace{1.5cm}


\begin{center}
{\huge Global Seiberg--Witten maps for \\[.2em] $U(n)$-bundles on tori and T-duality}

\end{center}


\vspace{0.4cm}

\begin{center}
{\bf   Paolo Aschieri,$\!^{1,2,3}$ Andreas Deser$\:\!^2$  }
 
\end{center}


\vspace{0.4cm}

\begin{center}
$^1$\small{Dipartimento di Scienze  e Innovazione Tecnologica \\
Universit\`a of Piemonte Orientale \\
Viale T. Michel 11, 15121, Alessandria }

  \vspace{0.2cm}

  $^2$\small{Istituto Nazionale di Fisica Nucleare, Sezione di Torino} \\ 
  { Via Pietro Giuria 1, 10125 Torino }\\

  \vspace{0.2cm}

$^3$\small{Arnold--Regge centre, Torino, via P. Giuria 1, 10125, Torino, Italy}
  \vspace{0.4cm}

\small{{\tt paolo.aschieri@uniupo.it}, {\tt deser@to.infn.it} }

\vspace{0.25cm}

\end{center} 

\vspace{0.4cm}

\begin{center}
\today
\end{center}

\begin{abstract}
Seiberg--Witten maps are a well-established method to locally
construct noncommutative gauge theories starting from commutative
gauge theories. We revisit and classify the ambiguities and the freedom in the
definition. 

 Geometrically, Seiberg--Witten maps provide
a quantization of bundles with connections. We study the case of
$U(n)$-vector bundles on two-dimensional tori, prove the
existence of globally defined Seiberg--Witten maps (induced from the
plane to the torus)
and show their compatibility with Morita equivalence. 

\end{abstract}
\sk\sk\sk $~~~~~~~~$\small{Keywords: {Deformation quantization,
    Ambiguities in Seiberg--Witten map,
      noncommutative\\$~~~~~~~~~~~~~~~~~~~~~~$ torus, noncommutative
    vector bundles, projective modules,  Morita equivalence}}

\clearpage

\tableofcontents


\section{Introduction}

The geometry of noncommutative tori is one of the most studied and
inspiring examples in the mathematics and physics literature.
Yang-Mills theory on 
noncommutative tori was defined in \cite{Connes:1987ue} and in \cite{Connes:1997cr} it was shown that noncommutative
tori emerge as backgrounds for compactifications of M-theory.
The study of Yang-Mills and Born-Infeld theories on noncommutative
tori has proven very fruitful:  On the one hand it allows to realize
M-theory and string theory
duality transformations within the low energy physics of Noncommutative
(Super) Yang-Mills theories \cite{Connes:1997cr}. 
On the other hand it provides exact low energy D-branes effective actions (in a given
$\alpha' \to 0$ sector of string theory where closed strings
decouple). This is a general feature of low energy effective actions
of open strings in the presence of a constant background flux ($B$-field) and has led to
the Seiberg--Witten map between commutative and noncommutative gauge
fields \cite{Seiberg:1999vs}.  This is a field transformation that
allows to rewrite a gauge theory on commutative space as a gauge
theory on noncommutative space. Particularly relevant are then action
functionals, like the Yang-Mills one, that are invariant in form under
this map (background independent). 

Most of the literature on Seiberg--Witten map does not consider global
geometric aspects and focusses on the local properties. Furthermore
this map has mainly been studied in the context of formal deformation
quantization. 
An interesting result there concerns abelian gauge theories,  in that case the Seiberg--Witten map has been generalized to nonconstant background
$B$-fields using Kontsevich formality theorem \cite{Jurco:2001my},
and shown to quantize line bundles on a Poisson manifold to quantum
line bundles  \cite{Jurco:2001kp}.

\sk
In this paper we study global aspects of the Seiberg--Witten map for
nonabelian $U(n)$-gauge fields on noncommutative tori and show that it
quantizes vector bundles on tori with connections to vector
bundles on noncommutative tori with noncommutative connections, these are the nonformal (Hilbert)
modules over noncommutative tori studied in the mathematics
and physics literature \cite{rieffel1981}, \cite{Connes:1997cr},
\cite{Ho:1998hqa}, \cite{Morariu:1998qm}, \cite{Schwarz:1998qj}. 
In this global nonformal context we also revisit the relation between
Seiberg--Witten maps and  Morita equivalence (T-duality)
transformations showing their compatibility.

Since the torus is a quotient of the plane, gauge and matter fields
associated with a bundle on the torus are just gauge and matter fields
(of trivial bundles) on the plane that satisfy twisted periodicity
conditions, determined by $U(n)$-valued functions $\Omega_\al$ on the plane.
While these are thought as ``transition functions'' of the torus, from
the plane perspective they are endomorphisms of the trivial rank $n$ bundle on
the plane, or equivalently, sections of the trivial rank $n\times n$
bundle, transforming in the adjoint representation.
The basic idea is then that the Seiberg--Witten map quantization of
these sections defines the quantum bundle on the
noncommutative torus via noncommutative twisted periodicity condition.
In this way the Seiberg--Witten map on the plane induces a
quantization of bundles on tori.  Consistently we show that the
Seiberg--Witten map for matter fields on the plane induces a
Seiberg-Witten map for matter fields on the torus.

A key point in order to obtain as explicit solutions the quantum
sections $\phi^\theta$ studied via different methods in \cite{Ho:1998hqa} is to exploit the freedom
in the Seiberg--Witten differential equation defining the
Seiberg--Witten map. This has led us to an exaustive study of the
ambiguities in the definition of the Seiberg--Witten map extending
previous results in \cite{Asakawa:1999cu} and \cite{Suo:2001ih}.

\sk

The paper is organized as follows. We begin by reviewing
the Seiberg--Witten map on $\mathbb{R}^d$, including the original Seiberg--Witten
differential equations and their recursive solutions as
formal power series in the noncommutativity parameter $\theta$.  
We present a simple way to classify the ambiguities of the Seiberg--Witten
map by classifiyng all terms allowed in the  Seiberg--Witten
equations, including those breaking covariance under constant $GL(d,\mathbb{R})$
rotations on $\mathbb{R}^d$, that is anyhow broken on the
$d$-dimentional torus. 

In chapter three we present
quantized $U(n)$-bundles on tori following
\cite{Ho:1998hqa} and establish the 
relation to the projective modules description of the more mathematical
literature, where $\theta$ is a real number.
We also describe endomorphisms of these  
bundles, and connections. All these fields are seen as functions on the
plane satisfying twisted periodicity conditions.

In the
following chapter four we prove the main result, the
compatibility of the Seiberg--Witten map with the twisted periodicity
conditions defining bundles with connections on the torus, and show
that the quantized sections are obtained via the
Seiberg--Witten map. This establishes a global and converging
Seiberg--Witten map for $U(n)$-bundles on the torus. In particular
different ordering prescriptions for the quantization of the algebra of
sections in the adjoint correspond to different choices of
Seiberg--Witten map.
 In chapter five we review Morita equivalence and its
T-duality transformations and show its compatibility with the Seiberg--Witten map.

\section{Seiberg--Witten map on $\mathbb{R}^d$}\label{SWl}

In this section we recall the Seiberg--Witten map between  the
noncommutative gauge fields $\SWA$ and gauge parameters $\epsihat$  on the noncommutative
$n$-dimensional plane 
$\mathbb{R}^d_\theta$ and  the ordinary ones $A$ and $\epsi$ on 
$\mathbb{R}^d$. The $n$-dimensional plane 
$\mathbb{R}^d_\theta$ is described by the noncommutative
algebra of formal power series in $\theta$ with coefficients in complex valued smooth functions on $\mathbb{R}^d$.
Noncommutativity is given by the Moyal-Weyl star
product, with the following conventions:
\begin{equation}
  (f\star g)(\sigma) =\, \exp \Bigl(i\pi\theta^{\mu\nu}\frac{\partial}{\partial \sigma^\mu} \frac{\partial}{\partial \rho^\nu} \Bigr)\,f(\sigma)g(\rho)|_{\sigma =\rho}\;,
\end{equation}
so that $[\sigma^\mu,\sigma^\nu]_\star = 2\pi i\theta^{\mu\nu}$ (with $\theta^{\mu\nu} = -\theta^{\nu\mu}$). We adopt conventions typically used in the literature on tori, which differs from conventions used for the noncommutative plane by a factor of $2\pi$ (i.e. $\vartheta^{\mu\nu} = 2\pi \theta^{\mu\nu}$ with $\vartheta^{\mu\nu}$ the noncommutativity parameter of \cite{Seiberg:1999vs}).
The Seiberg--Witten map relates the noncommutative gauge
potential $\SWA$  and the noncommutative gauge
parameters $\epsihat$ to  the ordinary $A$ and $\epsi$ so as to satisfy:
 \eq{
 \Ahat (A + \de_\epsi A) = \Ahat (A) + \dehat_\epsihat \Ahat (A) \label{SWcondition}
 }
 with 
  \eqa 
   & &
  \de_\epsi A_\mu = \part_\mu \epsi - i A_\mu 
      \epsi+ i \epsi A_\mu, \\
      & &
  \dehat_\epsihat \Ahat_\mu = \part_\mu \epsihat - i \Ahat_\mu \star     
      \epsihat + i \epsihat \star \Ahat_\mu.\label{3.44h}
      \ena

\noindent Condition (\ref{SWcondition}) states that the dependence of the
noncommutative gauge field on the ordinary one is such that ordinary gauge variations of $A$ inside $\Ahat(A)$ produce the noncommutative
gauge variation of $\Ahat$.  In a gauge theory physical quantities are
gauge invariant: they do not depend on the gauge potential but on
the gauge equivalence class of the potential given. 
The Seiberg--Witten map relates the noncommutative gauge fields to the
commutative ones by requiring noncommutative fields  to have the same gauge equivalence
classes as the commutative ones. 
Equation  (\ref{SWcondition}) can be solved order by order in $\theta$
yielding
$\Ahat$ and $\epsihat$ as  power series in $\theta$:  
 \eqa 
    \Ahat (A) &=& A + A^1 (A)  + A^2 (A) + \cdots + A^n (A)+ \cdots  \\
     \epsihat (\epsi, A)  &=&  \epsi + \epsi^1 (\epsi, A)+ \epsi^2 (\epsi, A)+ \cdots + 
     \epsi^n (\epsi, A)+ \cdots 
     \ena    
 \noi  where $A^n (A)$ and $\epsi^n (\epsi, A)$  are of order $n$ in $\theta$. Note that  $\epsihat$
depends on the ordinary $\epsi$ and also on $A$.

The Seiberg--Witten condition (\ref{SWcondition}) holds for any value of the
noncommutativity parameter $\theta$.
If we consider it at $\theta'$ and at $\theta$ we easily obtain 
that gauge equivalence classes of the $\theta'$-noncommutative theory
have to correspond to gauge equivalent classes of the
$\theta$-noncommutative theory, i.e., we generalize
(\ref{SWcondition}) to
\begin{equation}\label{SWconditionp}
\SWAp_\kappa(\SWA+\SWV_{\SWE} \SWA)=\SWAp_\kappa(\SWA) -
\SWVp_{\!\!\;\SWEp}\SWAp_{\kappa\:\!}(\SWA)~,
\end{equation}
where we denoted by $\star'$, $\SWAp$, $\SWEp$, $\SWVp_{\!\!\:\SWEp}$
the star product, the gauge potential, the gauge parameter and the gauge variation: 
$\SWVp_{\!\!\;\SWEp}\SWAp_\kappa=\partial_\kappa\SWEp-i\SWAp_\kappa\star'\SWEp+i\SWEp\star'\SWAp_\kappa$ at
noncommutativity parameter $\theta'$.
By considering $\theta$ and $\theta'$ infinitesimally close, so that $\theta' = \theta + \delta \theta$ and $\SWAp = \SWA + \delta \theta^{\mu\nu} \frac{\partial \SWA}{\partial \theta^{\mu\nu}}$ (we use the convention $\frac{\partial}{\partial \theta^{\mu\nu}}$ independent from $\frac{\partial}{\partial \theta^{\nu\mu}}$ and hence we sum over all $\mu$,$\nu$ indices), in \cite[\S 3.1]{Seiberg:1999vs} it is shown that if $\Ahat$ and
$\epsihat$  solve the differential equations
   \eqa
 & & { \part \over \part \theta^{\mu\nu}} \Ahat_\kappa = -
 \frac{\pi}{4} \Big(\{ \Ahat_{\mu}, \part_{\nu } \Ahat_\kappa +
 \Fhat_{\nu\kappa} \}_\star    +   \{ \Ahat_{\nu}, \part_{\mu } \Ahat_\kappa +
 \Fhat_{\mu\kappa} \}_\star \Big)   \label{one} \\
 & &  { \part \over \part \theta^{\mu\nu}} \epsihat =  -
 \frac{\pi}{4} \Big(\{  \Ahat_{\mu}, \part_{\nu} \epsihat \}_\star  +  \{  \Ahat_{\nu}, \part_{\mu} \epsihat \}_\star\label{two}\Big)
 \ena
  with the definitions
  \eqa
  & &
\Fhat_{\mu\nu} := \part_\mu \Ahat_\nu -  \part_\nu \Ahat_\mu - i \Ahat_\mu \star \Ahat_\nu
    + i \Ahat_\nu \star \Ahat_\mu \label{Fhatmunu} \\
   & &
    \{f,g\}_\star := f \star g + g \star f
    \ena
then $\SWAp(\SWA)$ and $\SWEp(\SWE,\SWA)$
satisfy also the Seiberg--Witten condition 
(\ref{SWconditionp}) for arbitrary  values of $\theta'$ and $\theta$, and in
particular, therefore, solve the Seiberg--Witten condition (\ref{SWcondition}).

   The differential equations (\ref{one}) and (\ref{two}) admit solutions in terms of formal power series in $\theta$, they are given recursively by \cite{Ulker:2007fm}
    \eqa
     & &
     A^{n+1}_\mu = -{\pi \over 2(n+1)} \theta^{\rho\sigma} \{\Ahat_\rho, \part_\sigma \Ahat_\mu   
     + \Fhat_{\sigma\mu} \}^n_\star\,, \label{An+1}\\
      & & \epsi^{n+1}=  -{\pi \over 2(n+1)} \theta^{\rho\sigma} \{\Ahat_\rho, \part_\sigma \epsihat \}^n_\star\,,\label{rtwo}
       \ena
   \noi where  $ \{ \fhat, \ghat \}^n_\star $ is the $n$-th  order term in $ \{ \fhat, \ghat \}_\star $, so that for example
     \eq{
     \{\Ahat_\rho, \part_\sigma \epsihat \}^n_\star \equiv \sum_{r+s+t=n} (A^r_\rho \star^s \part_\sigma \epsi^t +
      \part_\sigma \epsi^t \star^s A^r_\rho )\,,
       }
    \noi here $\star^s$ indicates the $s$-th order term in the star
    product expansion.\footnote{There is a  simple proof of
    \eqref{An+1}, \eqref{rtwo} \cite{Aschieri:2011ng}:
   multiplying the differential equations by $\theta^{\mu\nu}$ and analysing them
   order by order yields
   \eqa & &\theta^{\mu\nu}  { \part \over \part \theta^{\mu\nu}} A_\rho^{n+1} = (n+1) A_\rho^{n+1}= - \frac{\pi}{2}  \theta^{\mu\nu}  \{ \Ahat_{\mu}, \part_{\nu } \Ahat_\rho + \Fhat_{\nu\rho} \}_\star^n~,~~\nonumber\\
   & &\theta^{\mu\nu}  { \part \over \part \theta^{\mu\nu}} \epsi^{n+1}=  (n+1) \epsi^{n+1}= - \frac{\pi}{2} \theta^{\mu\nu}  \{  \Ahat_{\mu}, \part_{\nu} \epsihat \}_\star^n\nonumber
  \ena
\noindent since $A_\rho^{n+1}$ and $\epsi^{n+1}$ are homogeneous functions of $\theta$ of order 
$n+1$. }
\sk
Similar considerations hold
for matter fields $\phi$  transforming
in the fundamental or in the adjoint representation of the gauge group.
The Seiberg--Witten condition reads \cite{Jurco:2001rq}
\begin{equation}\label{SWconditionmatter}
\phihat(A+\delta_\epsi 
A, \phi+\delta_\epsi \phi) =\phihat(A,\phi)+\hat\delta_\epsihat\phihat(A,\phi)\,,
\end{equation}
or more generally,
\begin{equation}\label{SWconditionmatterp}
\SWPp(\SWA+\delta_{\SWE} \SWA, \SWP+\delta_{\SWE} \SWP) =\SWPp(\SWA,
\SWP)+\SWVp_{\SWEp}\SWPp(\SWA,\SWP)\,,
\end{equation}
and it is satisfied if the matter fields solve the differential equation
\eq{\label{matterfund}
 \delta \theta^{\mu\nu} \frac{\partial \SWP}{\partial \theta^{\mu\nu}}
 =-\frac{\pi}{2} \delta\theta^{\mu\nu}\SWA_\mu \star
 (\partial_\nu \SWP + D_\nu \SWP )\; ~~~~~~~~~~ {\rm
   fundamental ~rep.,~ i.e., ~}\dehat_\epsihat \phihat = i \epsihat \star \phihat\,, ~~~~
}
\eq{\label{matteradj}
 \delta \theta^{\mu\nu} \frac{\partial \SWPA}{\partial \theta^{\mu\nu}}
 &=-\frac{\pi}{2} \delta\theta^{\mu\nu}\bigl\{\SWA_\mu \,, (\partial_\nu
 \SWPA + D_\nu \SWPA )\bigr\}_\star\;  ~~~~~~ {\rm
   adjoint ~rep.,~ i.e.,~ }\dehat_\epsihat \Psihat = i \epsihat \star \Psihat - i \Psihat \star \epsihat\,. 
}
The explicit solutions order by order in $\theta$ are
\eqa     & & \phi^{n+1}= -{\pi \over 2(n+1)} \theta^{\mu\nu} \left(\Ahat_\mu \star (\part_\nu \phihat
     + D_\nu \phihat) \right)^n ~~~~~~~~~~{\rm (fundamental)}\label{phirec} \\
 & & \Psi^{n+1}= -{\pi\over 2(n+1)} \theta^{\mu\nu} \{\Ahat_\mu, \part_\nu \Psihat
     + D_\nu \Psihat \}^n_\star ~~~~~~~~~~~~~~{\rm (adjoint)}
\label{phiadrec}
\ena
where
\eq{D_\nu \phihat = \part_\nu \phihat -i  \Ahat_\nu \star \phihat ~~,~~~~ 
D_\nu \Psihat = \part_\nu \Psihat -i  \Ahat_\nu \star \Psihat +
i \Psihat \star \Ahat_\nu}
are the covariant derivative in the fundamental 
and  in the adjoint.

\subsection{Ambiguities in the Seiberg--Witten map}\label{sec:ambig}

The solution to the Seiberg--Witten conditions (\ref{SWcondition}), (\ref{SWconditionmatter}) is not unique. For example
      if $\hat A_\mu$ is a solution, any  noncommutative gauge
      transformation of $\hat A_\mu$ gives another solution. Another
      source of ambiguities is that of field redefinitions of the
      gauge potential (e.g., if $\hat A_\mu$ is a solution then so is
      $\hat A_\mu+\theta^{\rho\sigma}\theta^{\lambda\eta}\hat F_{\rho
        \lambda}\star D_\sigma \hat F_{\eta\mu}$).
We present a novel study of these ambiguities by considering the
freedom in modifying the differential equations (\ref{one}),
(\ref{two}) and (\ref{matterfund}) leading to the Seiberg--Witten
conditions.

We generalize the 
Seiberg--Witten equations allowing for three extra terms
$\hat D_{\mu\nu\rho}(\SWA),  \hat E_{\mu\nu}(\SWE,\SWA)$ and
$\hat C_{\mu\nu}(\SWP,\SWA)$ or $\hat C_{\mu\nu}(\SWPA,\SWA)$ (for the
fundamental or adjoint representation) that are a priori arbitrary
functions of their arguments and derivatives thereof, that are (formal) power series in $\theta$ and 
that are antisymmetric in the $\mu,\nu$
indices. We hence consider the equations
\begin{eqnarray}
  \label{extendedSWA}
  \!\!\!\!\!\delta^\theta \SWA_\kappa\,=\, \delta \theta^{\mu\nu} \frac{\partial
  \SWA_\kappa}{\partial \theta^{\mu\nu}}&\!\!
                                          =&\!\!\frac{\pi}{2}\delta\theta^{\mu\nu}\Bigl(\SWA_\mu\star(\partial_\nu
                                             \SWA_\kappa +
                                             \SWF_{\nu\kappa}) +
                                             (\partial_\nu \SWA_\kappa
                                             + \SWF_{\nu\kappa}) \star
                                             \SWA_\mu +
                                             \hat D_{\mu\nu\kappa}(\SWA)\Bigr)\,,~~~~~~\\
\label{extendedSWE}
   \!\!\!\!\! \delta^\theta\SWE\,=\,\delta \theta^{\mu\nu}
  \frac{\partial\SWE}{\partial \theta^{\mu\nu}} &\!\!=&\!\!\frac{\pi}{2}
                                                      \delta\theta^{\mu\nu}\Bigl(\partial_\mu
                                                        \SWE \star
                                                        \SWA_\nu +
                                                        \SWA_\nu
                                                        \star \partial_\mu
                                                        \SWE +
                                                       \hat  E_{\mu\nu}(\SWE,\SWA)\Bigr)\,,\\
\label{extendedSWP}
    \!\!\!\!\!  \delta^\theta\SWP\,=\,\delta \theta^{\mu\nu} \frac{\partial
  \SWP}{\partial \theta^{\mu\nu}} &\!\!=&\!\!-\frac{\pi}{2}
                                          \delta\theta^{\mu\nu}\Bigl(\SWA_\mu
                                          \star \partial_\nu \SWP +
                                          \SWA_\mu \star D_\nu \SWP +
                                          \hat
                                          C_{\mu\nu}(\SWP,\SWA)\Bigr)\,,\\
   \!\!\!\!\!   \delta^\theta\SWPA\,=\,\delta \theta^{\mu\nu} \frac{\partial
  \SWPA}{\partial \theta^{\mu\nu}} &\!\!=&\!\!-\frac{\pi}{2}
                                          \delta\theta^{\mu\nu}\Bigl(\SWA_\mu
                                          \star (\partial_\nu \SWPA +
                                          D_\nu \SWPA) + (\partial_\nu \SWPA +
                                          D_\nu \SWPA) \star \SWA_\mu +
                                          \hat
                                           C_{\mu\nu}(\SWPA,\SWA)\Bigr)\nonumber\\
\label{extendedSWPA}
\end{eqnarray}
and observe that $\hat E_{\mu\nu}$ must be linear in $\SWE$ since all terms in
(\ref{extendedSWE}) but $\hat E_{\mu\nu}$ are linear in $\SWE$,  similarly
$\hat C_{\mu\nu}$ must be linear in $\SWP$ because of the linearity in
$\SWP$ of all other terms in (\ref{extendedSWP}), and similarly for
$\hat C_{\mu\nu}(\SWPA,\SWA)$ in \eqref{extendedSWPA}. We further constrain the dependence of 
$\hat D_{\mu\nu\kappa}, \hat  E_{\mu\nu}$, and
$\hat C_{\mu\nu}$ on $\theta$, $\SWA$, $\SWE$ and their derivatives 
by requiring the Seiberg--Witten conditions  (\ref{SWconditionp}),
(\ref{SWconditionmatterp}). 
Let's first consider condition (\ref{SWconditionp}), we use  (\ref{extendedSWA}) in
$\SWAp(\SWA+\SWV_{\SWE} \SWA)=\SWA+\SWV_{\SWE} \SWA+\delta^\theta\SWA_{\,}(\SWA+\SWV_{\SWE} \SWA)$
as well as in $\SWAp(\SWA)=\SWA+\delta^\theta \SWA$, then we use 
$f\star' h=f\star h+i\pi\delta\theta^{\mu\nu} \partial_\mu
f\star\partial_\nu h+{\cal O}(\theta^2)$ and (\ref{extendedSWE})
in $\SWVp_{\!\!\;\SWEp}\SWAp_\kappa=\partial_\kappa\SWEp-i\SWAp_\kappa\star'
\SWEp+i\SWEp\star'\SWAp_\kappa$, 
and finally recall that for $\hat D_{\mu\nu\kappa}=\hat E_{\mu\nu}=0$ the equation is
satisfied, we thus obtain the condition
\begin{equation}\label{constrD}
\hat D_{\mu\nu\kappa}(\SWA+\SWV_{\SWE} \SWA)-\hat
D_{\mu\nu\kappa}(\SWA) - i[\SWE,\hat D_{\mu\nu\kappa}(\SWA)]_{\star}
=-  D_\kappa \hat E_{\mu\nu}(\SWE,\SWA)\;.
\end{equation}
Similarly, the Seiberg--Witten condition (\ref{SWconditionmatterp})  
for matter fields in the fundamental and in the adjoint constraints $\hat C_{\mu\nu}$ to satisfy
\begin{equation}\label{constrC}
\begin{matrix}&\hat C_{\mu\nu}(\SWA+\SWV_{\SWE}\SWA,\SWP+\SWV_{\SWE}\SWP) -
 \hat C_{\mu\nu}(\SWA,\SWP) - i\SWE\star \hat C_{\mu\nu}(\SWA,\SWP) =
 -i\hat E_{\mu\nu}(\SWE,\SWA)\star \SWP\;,\\[.8em]
&\hat C_{\mu\nu}(\SWA+\SWV_{\SWE}\SWA,\SWPA+\SWV_{\SWE}\SWPA) -
 \hat C_{\mu\nu}(\SWA,\SWPA)  -i[\SWE,\hat C_{\mu\nu}(\SWA,\SWPA)]_\star =
  -i[\hat E_{\mu\nu}(\SWE,\SWA),\SWPA]_\star\;,
\end{matrix}
\end{equation}
where $[f,g]_\star= f\star g-g\star f$ is the star commutator.

An immediate comment is that any $\hat D_{\mu\nu\kappa}$ and $\hat C_{\mu\nu}$ covariant under gauge transformations solve (\ref{constrD}), (\ref{constrC}) with
$\hat E_{\mu\nu}=0$.

\sk 
In summary, we have shown that the most general solution
$\SWA(A)$, $\SWE(\varepsilon,A)$, $\SWP(A,\phi)$, $\SWPA(A,\Psi)$ of the Seiberg--Witten
conditions (\ref{SWcondition}), (\ref{SWconditionmatter}) is given by the differential equations 
(\ref{extendedSWA})-(\ref{extendedSWPA}) where $\hat D,\hat E, \hat C$
are constrained by (\ref{constrD}) and (\ref{constrC}).
Further constraints on the $\hat D,\hat E, \hat C$ terms are obtained by requiring that Seiberg--Witten map respects hermiticity and charge conjugation in the sense that
the hermiticity and charge conjugation properties of the commutative
fields imply those of the noncommutative fields \cite{Aschieri:2002mc, Aschieri:2011ng}.
\sk
It is useful to compare these results with the previous ones in the
literature. 
To this aim we constrain $\hat D,\hat E, \hat C$
to be $\star$-polynomials in
the variables  $\SWA_\mu, D_\mu ,\SWE, \SWP,  \SWPA,
\theta^{\mu\nu}, \partial_\mu$, that is in the variables
$\SWA_\mu, D_\mu ,\SWE, \SWP,  \SWPA, \theta^{\mu\nu}$,
indeed the partial derivatives $\partial_\mu$ can be always expressed
in terms of the covariant derivatives and the gauge potentials. We then define
the operator $\CL_{\SWE}$ to satisfy the Leinbiz rule and to be the adjoint action of $\SWE$ on fields
in the adjoint, and the fundamental action otherwise: 
$\CL_{\SWE}(A_\mu)=[\SWE,\SWA_\mu]_\star~,~
\CL_{\SWE}(D_\mu)=[\SWE, D_\mu]_\star~,~
\CL_{\SWE}(\SWE)=
0\,,~ 
\CL_{\SWE}(\SWP)=\SWE\star\SWP\,,\CL_{\SWE}(\SWPA)=[\SWE,\SWPA]_\star
$.
Next, as in \cite{Asakawa:1999cu}, we introduce the operator 
\eq{\delta'_{\SWE}=\partial_\mu\SWE\frac{\partial}{\partial\SWA_\mu} }
that acts only on the gauge
potential $\SWA_\mu$  and does not act on the
covariant derivatives: $\delta'_{\SWE} (D_\mu)=0$. This operator
just substitutes $\SWA_\mu$ with $\partial_\mu\SWE$ and 
satisfies the Leibniz rule. We can consider (\ref{constrD})-(\ref{constrC}) as a combinatorial problem in the words (symbols)
$\SWA_\mu,D_\mu, \SWP, \SWPA,\SWE, \theta^{\mu\nu}$. Since 
$\delta'_{\SWE}=\SWV_{\SWE}-i\CL_{\SWE}$ (as is easily seen on the generators
$\SWA_\mu,D_\mu, \SWP,\SWE,  \theta^{\mu\nu}$, and then
recalling the Leibniz rule) we can rewrite  (\ref{constrD}) and
(\ref{constrC}) as
\begin{equation}\label{constrD'}
\delta'_{\SWE}\hat D_{\mu\nu\kappa}(\SWA)=-  D_\kappa \hat E_{\mu\nu}(\SWE,\SWA)\;,
\end{equation}
\begin{equation}\label{constrC'}
\delta'_{\SWE}\hat C_{\mu\nu}(\SWA, \SWP)=
 -i\hat E_{\mu\nu}(\SWE,\SWA)\star \SWP~~,~~~
\delta'_{\SWE}\hat C_{\mu\nu}(\SWA, \SWPA)=
 -i[\hat E_{\mu\nu}(\SWE,\SWA), \SWPA]_\star\;.
\end{equation}
If we constrain  $\hat D,\hat E, \hat C$ to be words
($\star$-polynomials)  only of $\SWA_\mu,D_\mu, \SWP, \SWPA, \SWE$ and
not of the other letters $\theta^{\mu\nu}$ we then have a finite
number of letters that can match the
dimensionality of $\hat D,\hat E, \hat C$ and the linearity
constraints on $\SWE$, $\SWP$ and $\SWPA$ discussed immediately after
(\ref{extendedSWPA}). In $\mathbb{R}^d$ the star product 
$f\star g$  is invariant under constant
$GL(d,\mathbb{R})$ coordinate transformations; if we do not fix the gauge potential $\SWA_\mu$ it is natural to 
implement this symmetry also in the Seiberg--Witten map, i.e., to ask
all  expressions to be tensorial under constant coordinate
transformations.
Since $\theta^{\mu\nu}$, $\SWA_\mu$, $\partial_\mu$ and $D_\mu$
transform covariantly, this is achieved, as usual, by contracting indices
tensorially and by  matching the index structure  of $\hat D,\hat E,
\hat C$.

We have written the most general linear combination of letters for
$\hat D$,  for $\hat E$ and $\hat C$, substituted them in (\ref{constrD'}) and
(\ref{constrC'}) and shown (using also Bianchi identities that are 
indeed combinatorial identities) that the most general
$GL(d,\mathbb{R})$ covariant solution
with no explicit dependence on $\theta^{\mu\nu}$ is 
\begin{align}\label{DEsemplice}
&\hat D_{\mu\nu\kappa}=\alpha D_\kappa \SWF_{\mu\nu}+\beta
D_\kappa[\SWA_\mu,\SWA_\nu]_\star~~,~~~~\hat E_{\mu\nu}=2\beta [\partial_\mu\SWE,\SWA_\nu]_\star\;,\\
&\label{Cfund}\hat C_{\mu\nu} = -2i\beta[\SWA_\mu,\SWA_\nu]_\star \star \SWP + \gamma
  \SWF_{\mu\nu}\star \SWP 
~~~~~~~~~~~~~~~~~~~~~~~~~~~~~~~~~~{\rm (fundamental)}\;
\\
&\label{Cadj}\hat C_{\mu\nu} = -2i\beta[[\SWA_\mu,\SWA_\nu]_\star, \SWPA]_\star + \gamma'
  \SWF_{\mu\nu}\star \SWPA +\tilde\gamma \SWPA \star \SWF_{\mu\nu}
~~~~~~~~~~\:~~~~~~ {\rm (adjoint)}\;
\end{align}
with $\alpha,\beta,\gamma, \gamma'$ and $\tilde\gamma$ arbitrary constants (terms like
$D_\mu D_\nu \SWP$ are proportional to $F_{\mu\nu}$ due to
antisymmetry in $\mu,\nu$).
This shows that the results in \cite{Asakawa:1999cu} and in \cite{Suo:2001ih}
are the most general solutions covariant under constant $GL(d,\mathbb{R})$ coordinate transformations and that do not depend explicitly on $\theta^{\mu\nu}$.
\sk

If we relax the covariance under  constant
$GL(d,\mathbb{R})$ coordinate transformations we obtain further
solutions, in particular we can generalize \eqref{Cfund} to
\eq{\hat C_{\mu\nu} = -2i\beta[\SWA_\mu,\SWA_\nu]_\star \star \SWP + \gamma
  \SWF_{\mu\nu}\star \SWP + \rho D_\mu D_\nu \SWP
  ~\mbox{ { for $\mu<\nu$,}
  and }~ \hat C_{\nu\mu}:=-\hat C_{\mu\nu}\label{CDD}}
(and similarly for the adjoint case). Here $\rho$ is a constant, and
$GL(d,\mathbb{R})$ covariance is broken because  $\hat C_{\nu\mu}$ has
not the term $\rho D_\nu D_\mu\phi$. Actually, as noted after \eqref{constrC}, we can add to $\hat C_{\mu\nu}$ an
arbitrary term covariant under gauge transformations, like e.g. 
the linear combination $\eta D_\mu D_\mu \SWP+\omega D_\nu D_\nu \SWP$ that even more clearly
breaks $GL(d,\mathbb{R})$ coordinate transformations. When considering
Seiberg--Witten maps on the $d$-dimensional torus 
$GL(d,\mathbb{R})$ is anyhow broken, and  $\hat C_{\mu\nu}$ terms like those
listed can be considered, and, as we will see, are important in order to
explicitly solve the Seiberg-Witten differential equations.
\sk
We further remark that it is very natural to allow for polynomials that depend
explicitly also on $\theta$ (for example the field redefinition 
$\SWA_\mu\to \hat A_\mu+\theta^{\rho\sigma}\theta^{\lambda\eta}\hat
F_{\rho  \lambda}\star D_\sigma \hat F_{\eta\mu}$ is of this kind). In this
case (\ref{DEsemplice}) is not the most general solution.
As an example consider 
$
\hat D_{\mu\nu\kappa}=\gamma(\theta^{\rho\sigma}\SWF_{\rho\sigma})^pD_\kappa\SWF_{\mu\nu}$
(with $\gamma$ arbitrary constant, $p$ integer) and $\hat E=0$.

\section{Noncommutative tori}
Noncommutative tori are among the most studied objects in
noncommutative geometry. In physics, they serve as key examples to
study T-duality. We review the classification of bundles  (finitely
generated projective modules) with connections on noncommutative tori. 

\subsection{$U(n)$-vector bundles on commutative tori}\label{UNCT}
We consider hermitian $n$-dimensional vector bundles over the
2-dimensional torus $\TT2$, i.e.,
rank $n$ complex vector bundles canonically associated (via the
fundamental representation) to a $U(n)$-principal bundle. For short we
will call these bundles $U(n)$-vector bundles or simply $U(n)$-bundles. 

While a usual description of bundles is via local sections defined on
opens and transition functions on overlaps, since the torus $\TT2$ is given by the quotient
$\mathbb{R}^2/(2\pi\mathbb{Z})^2$, bundles on $\TT2$ are
most easily described by sections of the trivial
vector bundle on the plane 
$\mathbb{R}^2\times\mathbb{C}^n\to \mathbb{R}^2$ that obey
twisted periodicity conditions, often called boundary
conditions. These conditions can be seen as arising from the 
transition functions of the bundle on $\TT2$ in the limiting case of
opens that overlap only on the boundary of the fundamental domain
determining the torus as a quotient of $\mathbb{R}^2$. 
Let's describe the smooth sections
${\cal E}_{n,m}$ of a $U(n)$-vector bundle with topological charge $m$
\cite{tHooft:1981nnx} (we follow \cite{Taylor:1997dy},
\cite{Ho:1998hqa}, see also \cite{vanBaal:1982ag}).
Define the fundamental domain in $\mathbb{R}^2$ to
be the square of length $2\pi$, so that $\TT2= \mathbb{R}^2/(2\pi
\mathbb{Z})^2$ and functions on $\TT2$ are $2\pi$-periodic
functions on $ \mathbb{R}^2$,
and define the $U(n)$-matrix valued functions (transition functions)
\eq{\label{Om12}\Omega_1(\sigma^2) = e^{im\sigma^2/n}U~,~~~\Omega_2(\sigma^1) = V~,}
where $U$ and $V$ are  the clock and shift $U(n)$-matrices with entries
$U_{kl} = e^{2\pi i km/n}\delta_{kl}, V_{kl}=\delta_{(k+1)l}$ for $k<n$, $V_{n,l}=\delta_{1l}$. 
Let $\phi$ be an $n$-dimensional vector of complex valued functions on
$\mathbb{R}^2$ (a section of the trivial bundle
$\mathbb{R}^2\times\mathbb{C}^n\to \mathbb{R}^2$), 
then the
twisted periodicity conditions defining the sections $\phi \in {\cal
  E}_{n,m}$ is the system of $2n$ equations in $\mathbb{R}^2$:
\eq{  \label{classbd}
    \phi(\sigma^1 + 2\pi,\sigma^2) =\,& \Omega_1(\sigma^2)\phi(\sigma^1,\sigma^2)\;,\\[.2em]
    \phi(\sigma^1,\sigma^2 + 2\pi)=\,&\Omega_2(\sigma^1)\phi(\sigma^1,\sigma^2)\;,
}
where the $U(n)$-matrix valued functions $\Omega_\al$ satisfy the cocycle condition
  \begin{equation}
    \label{classcy}
    \Omega_1(\sigma^2 + 2\pi)\Omega_2(\sigma^1)=\Omega_2(\sigma^1+2\pi)\Omega_1(\sigma^2)\;.
  \end{equation}
We remark that \eqref{classbd} and \eqref{classcy} are equations for functions
($\phi_k, {\Omega_1}_{kl} , {\Omega_2}_{kl})$ on $\mathbb{R}^2$. More
geometrically, $\phi$ are sections of $\mathbb{R}^2\times\mathbb{C}^n\to \mathbb{R}^2$ (the trivial $U(n)$-vector bundle on $\mathbb{R}^2$), and
$\Omega_\al:\mathbb{R}^2\to U(n)\subset M_{n\times n}(\mathbb{C})$, $\al=1,2$,
are endomorphisms of this vector bundle that transform sections to
sections $\phi\mapsto \Omega_\al\phi$ (they are sections of the
endomorphism bundle ${\rm End}(E)$). Denoting by
$E=C^\infty({\mathbb{R}^2})^{\,\oplus n}$ the module of sections of
$\mathbb{R}^2\times\mathbb{C}^n\to \mathbb{R}^2$ we write
$\Omega_\al\in \textrm{End}(E)=C^\infty({\mathbb{R}^2})^{\,\oplus
  (n\times n)}$, where in the last equality we used that endomorphisms
of a trivial (smooth)  $n$-dimensional vector bundle are just (smooth) maps from the base space to 
linear maps on the fibre. Thus $\Omega_\al$ are sections of the
complexified adjoint bundle, for short, sections in the adjoint.

An explicit description of the sections $\phi\in {\cal E}_{n,m}$
solving \eqref{classbd} 
was provided in \cite{Ganor:1996zk} and requires defining: 
\eq{\label{AB}
\AAA:=
\frac{m}{n}\Bigl(\frac{\sigma^2}{2\pi} + k + ns\Bigr) +j~,~~~
\BBB:=i\sigma^1~,
} 
then an arbitrary section $\phi=(\phi_k)_{k=1,...n}\in  {\cal E}_{n,m}$ is given by 
\begin{equation}
  \label{classec}
  \phi_k(\sigma^1,\sigma^2)= \underset{s\in\mathbb{Z}}{\sum}\overset{m}{\underset{j=1}{\sum}}\, e^{\AAA\BBB}\tilde \phi_j(\frac{n}{m}\AAA)\;,
\end{equation}
where $m$ is the topological charge, $k={1,\dots n}$ and $\tilde
\phi_j$ are $m$ arbitrary complex valued Schwartz functions on
$\mathbb{R}$. More elegantly they define a  Schwartz function
$\tilde\phi: \mathbb{R}\times \mathbb{Z}_m\to \mathbb{C}$. 

In an analogous manner, covariant derivatives $D_\mu$ on the module of
sections ${\cal E}_{n,m}$ are described by 
covariant derivatives of the trivial $U(n)$-bundle on $\mathbb{R}^2$
satisfying the appropriate periodicity:
\eq{\label{connperiodicity}
  D_\mu\vert_{(\sigma^1+2\pi,\sigma^2)} &=\,\Omega_1(\sigma^2)D_\mu\vert_{(\sigma^1,\sigma^2)}\Omega_1^{-1}(\sigma^2)\;,\\[.2em]
  D_\mu\vert_{(\sigma^1,\sigma^2+2\pi)} &=\,\Omega_2(\sigma^1)D_\mu\vert_{(\sigma^1,\sigma^2)}\Omega_2^{-1}(\sigma^2)\;.
}

A particular solution to these conditions, suitable for later
generalizations to the noncommutative torus, is the one with only a
single non-vanishing component of the gauge potential, proportional to
the unit $n\times n$-matrix:
\begin{equation}\label{CovDer}
  D_1 =\,\partial_1-iA_1=\,\partial_1\;,\quad D_2= \partial_2 - i\,A_2=\,\partial_2 - i\frac{m}{n} \frac{\sigma^1}{2\pi}\mathbf{1}\;. 
\end{equation}
The field strength is given by
$F=F_{12}= i[D_1,D_2]= \tfrac{1}{2\pi}\tfrac{m}{n}\mathbf{1}$ and indeed the topological charge of ${\cal
  E}_{n,m}$ is $m$: $\frac{1}{2\pi}\int \textrm{tr}(F) d\sigma^1d\sigma^2=m$.

The set of sections ${\cal E}_{n,m}$ is a module over the algebra ${\cal C}^\infty(\TT2)$ of smooth functions on $\TT2$, the action
$\phi\mapsto \phi f$ (it is customary to multiply functions from the
right rather than the left) is simply the product of the vector $\phi$
with the periodic function $f$. It is immediate to check that $\phi f$
satisfies \eqref{classbd} if so does $\phi$.\footnote{Since the 
bundle is a positive definite hermitian complex vector bundle and
the algebra of continuous functions $C(\TT2)$ is a $C^\star$-algebra we also have that ${\cal
  E}_{n,m}$ is a Hilbert module over $C(\TT2)$.}

Moreover, ${\cal
  E}_{n,m}$ is a module with respect to the algebra of endomorphisms
of ${\cal E}_{n,m}$ itself. By definition, an endomorphism is a
fiberwise linear map on the vector bundle that
acts as the identity on the base space (the torus), and hence it is a
linear map on sections that is the identity on functions on the torus: $\phi f\mapsto \Psi(\phi
f)=\Psi(\phi) f$. Therefore the $\textrm{End}({\cal E}_{n,m})$- and ${
  C}^\infty(\TT2)$-actions commute, and ${\cal E}_{n,m}$ is a bimodule with
respect to the algebras $\textrm{End}({\cal E}_{n,m})$ and ${
  C}^\infty(\TT2)$, we write this as
\begin{equation}\label{EndEMT}
  {\cal E}_{n,m} \in \; \prescript{}{\textrm{End}({\cal E}_{n,m})}{\mathcal{M}}_{{C}^\infty(\TT2)}\;.
\end{equation}
We conclude with an explicit description
of the algebra of endomorphisms of ${\cal
  E}_{n,m}$. First of all, the algebra of endomorphisms of 
$E$, the module of sections of $\mathbb{R}^2\times\mathbb{C}^n\to
\mathbb{R}^2$, as already observed, is given by all $n\times n$ matrix valued functions on
$\mathbb{R}^2$, $\textrm{End}(E)=\{ \Psi: \mathbb{R}^2\to M_{n\times
  n}(\mathbb{C})\}$; the action on sections is simply the matrix
transformation $\phi\mapsto \Psi\phi$. The
algebra $\textrm{End}({\cal{E}}_{n,m})$ is the subalgebra of
$\textrm{End}(E)$ that preserves the twisted periodicity conditions \eqref
{classbd}: if $\phi$ satisfies \eqref
{classbd} then so does $\Psi\phi$. That is, endomorphisms of ${\cal
  E}_{n,m} $ are endomorphisms of $E$ that satisfy the
twisted boundary conditions in the adjoint representation:
\eq{\label{endperiodicity}
  \Psi(\sigma^1+2\pi,\sigma^2)&=\,\Omega_1(\sigma^2)\Psi(\sigma^1,\sigma^2)\Omega_1^{-1}(\sigma^2)\;,\\[.2em]
  \Psi(\sigma^1,\sigma^2+2\pi) &=\,\Omega_2(\sigma^1)\Psi(\sigma^1,\sigma^2)\Omega_2^{-1}(\sigma^2)\;.
}
We see that they are the sections of the adjoint $U(n)$-vector bundle
on the torus (i.e., the complex vector bundle
canonically associated via the adjoint representation, rather than the
fundamental, to the $U(n)$-principal bundle).

For $m$ and $n$ coprime, the 
algebra $\textrm{End}({\cal E}_{n,m})$ is generated by the
$U(n)$-valued functions on $\mathbb{R}^2$ (cf. \cite{Ho:1998hqa},\cite{Morariu:1998qm}):
\begin{equation}\label{commgenerat}
  Z_1 =\,e^{i\sigma^1/n}V^b\;,\qquad Z_2=\, e^{i\sigma^2/n}U^{-b}\;,
\end{equation}
where $b\in \mathbb{Z}$ is such that $an-bm=1$ with  $a\in \mathbb{Z}$. (If $a'$ and $b'$ are another couple satisfying $a'n-b'm=1$, the algebra is
the same since $m,n$ coprime implies $(a-a')=ms,  (b-b')=ns, s\in
\mathbb{Z}$ and we have $U^n=V^n=1$). 
It is easy to see that $Z_1Z_2=e^{2\pi i b/n}Z_2Z_1$, henceforth
$\textrm{End}({\cal E}_{n,m})=T_{b/n}$ the algebra of the noncommutative
torus with rational noncommutativity parameter $b/n$.
\\

\subsection{$U(n)$-vector bundles on noncommutative tori}\label{UNNC}
The description of bundles on the torus via modules on the algebra $C^\infty(\TT2)$ of smooth functions on the torus, and the
description of these modules via vector valued functions $\phi$ on
$\mathbb{R}^2$ (sections of the trivial $U(n)$-vector bundle on
$\mathbb{R}^2$) satisfying twisted periodicity conditions determined by
the two matrix valued functions 
$\Omega_\al:\mathbb{R}^2\to U(n)$ (bundle endomorphisms) is
particularly well suited to noncommutative generalizations.

Consider as in Section \ref{SWl} the noncommutative Moyal-Weyl algebra
$\mathbb{R}^2_\theta$, with
$\theta=\theta^{12}=-\theta^{21}$,
\begin{equation}
  \label{starconv}
  (f\star g)(\sigma^1,\sigma^2) =\,fg(\sigma^1,\sigma^2) + i\pi\theta\bigl(\partial_{\sigma^1}f \,\partial_{\sigma^2} g - \partial_{\sigma^2} f \,\partial_{\sigma^1} g\bigr) + {\cal O}(\theta^2)~.
\end{equation}
The $\star$-product between $2\pi$-periodic functions on
the plane is again a  $2\pi$-periodic function, and therefore
$2\pi$-periodic functions form a subalgebra of $\mathbb{R}^2_\theta$, 
this is the noncommutative torus $T_{(-\theta)}$.
Explicitly, $T_{(-\theta)}$, for $\theta\in \mathbb{R}$, is defined as the algebra over $\mathbb{C}$ generated by two
invertible elements $U_1$, $U_2$ that satisfy the relations
\eq{U_1U_2=e^{2\pi i(-\theta)}U_2 U_1
}
with involution given by $U_1^\ast=U_1^{-1}$, $U_2^\ast=U_2^{-1}$. The
smooth noncommutative torus $T_{(-\theta)}$ is
$T_{(-\theta)}=\{\sum_{p,q\in \mathbb{Z}}a_{p,q}\,U_1^p U_2^q, a_{p,q}\in
\mathbb{C}\}$, where $a: \mathbb{Z}^2\to
\mathbb{C}\,,\;(p,q)\mapsto a_{p,q}$ are Schwarz functions on
$\mathbb{Z}^2$. We can also consider $\theta$ as a formal parameter (so
that $U_1$ and $U_2$ generate the algebra over the ring of formal
power series $\mathbb{C}[[\theta]]$).
Since in $\mathbb{R}^2_\theta$ \eq{e^{i\sigma^1}\star
  e^{i\sigma^2}=e^{2\pi i(-\theta)}e^{i\sigma^2}\star e^{i\sigma^1}~,} 
setting $U_1=e^{i\sigma^1}$, $U_2=e^{i\sigma^2}$ 
realizes this algebra as a subalgebra of $\mathbb{R}^2_\theta$. Notice that 
restricting to periodic functions allows to specialize 
the formal parameter $\theta$ to a real number.

\sk

The twisted boundary conditions \eqref{classbd} and the cocycle
conditions are equations for the functions ($\phi_k, {\Omega_1}_{kl} , {\Omega_2}_{kl}$) on $\mathbb{R}^2$ and deforming
the commutative product in the  $\star$-product we obtain the
noncommutative deformation of these conditions
\eq{
  \label{qperiodicity}
  \TSWp(\sigma^1+2\pi,\sigma^2) &= \;\Omega_1(\sigma^2)\star \TSWp(\sigma^1,\sigma^2)\;,\\[.2em]
  \TSWp(\sigma^1,\sigma^2+2\pi) &= \;\Omega_2(\sigma^1)\star \TSWp(\sigma^1,\sigma^2)\;,
}
\begin{equation}\label{Omperiodicity}
  \Omega_1(\sigma^2 + 2\pi)\star\Omega_2(\sigma^1) = \;\Omega_2(\sigma^1 + 2\pi)\star \Omega_1(\sigma^2)\;. 
\end{equation}
The solutions \eqref{Om12} of the commutative cocycle conditions are also solutions of the $\star$-cocycle condition.
The solutions $\TSWp$ to \eqref{qperiodicity} are immediately seen to
be a module with respect to the noncommutative torus subalgebra
$T_\theta\subset \mathbb{R}^2_\theta$: if and only if $f$ is a
periodic function we have that a $\TSWp$ satisfying
\eqref{qperiodicity} implies that  $\TSWp\star f$ satisfies
\eqref{qperiodicity}.
The solutions $\TSWp$ therefore span the module ${\cal E}^\theta_{n,m}$
of sections of a rank $n$ complex vector bundle on the noncommutative torus.

The explicit solution of \eqref{qperiodicity} requires the use of a normal ordered function
$E(\AAA,\BBB)$ defined by 
\begin{equation}
  \label{Efunction}
  E(\AAA,\BBB):=\,\frac{1}{1-[\AAA,\BBB]_{\star}}\,\overset{\infty}{\underset{l=0}\sum} \, \frac{1}{l!}\,\AAA^l\star \BBB^l\;,
\end{equation}
where the definition of $\AAA$ and $\BBB$ in terms of the coordinates
$\sigma^1,\sigma^2$ and the  integers $n$ and $m$ is
given in \eqref{AB}. 
For later work with the function  $E(\AAA,\BBB)$, we collect some of
its properties, whose proof follows for the relation
$[\AAA,\BBB]_\star=\frac{m}{n}\theta$ (the last two are easily derived
from the corresponding differential equation in $\lambda$).
\begin{lemma}\label{uno}
  The function $E(\AAA,\BBB)$ satisfies
  \begin{align}
    \AAA\star E(\AAA,\BBB) =\,&\frac{1}{1-c} E(\AAA,\BBB)\star \AAA\;,\\
    \BBB\star E(\AAA,\BBB) =\,&E(\AAA,\BBB)(1-c)\star \BBB\;,\\
E(-\BBB, \AAA)E(\AAA,&\BBB)=\,1\;,\\
E(\AAA+\lambda,\BBB)=\,&E(\AAA,\BBB)\star e^{\lambda B}\;,\label{EAlaB}\\
E(\AAA,\BBB+\lambda)=\,&e^{\lambda A} \star E(\AAA,\BBB)\;,
  \end{align}
  where we have set $c:=[\AAA,\BBB]_{\star} = \tfrac{m}{n}\theta$, and
  $\lambda\in \mathbb{C}$.
\end{lemma}
We can now recall the solution  $\Tphi=(\Tphi_k)^{}_{k=1,...n}$ presented in
\cite{Ho:1998hqa}, see \cite{Morariu:1998qm} for a derivation, to the twisted boundary conditions
\eqref{qperiodicity} 
\begin{equation}
  \label{qsec}
 \Tphi_k(\sigma^1,\sigma^2)=\,\underset{s\in\mathbb{Z}}{\sum}\overset{m}{\underset{j=1}{\sum}}\,E\Big(\tfrac{m}{n}\big(\tfrac{\sigma^2}{2\pi} + k + ns\big) +j,i\sigma^1\Big)\star \tilde \phi_j\big(\tfrac{\sigma^2}{2\pi} + k + ns +\frac{n}{m}j\big)\;,
\end{equation}
that for short we rewrite as 
\eq{\label{shand}
\Tphi_k(\sigma^1,\sigma^2)=\,\sum_{s\in \mathbb{Z}}\sum^m_{j=1}\,E(\AAA,\BBB)\star \tilde \phi_j(\tfrac{n}{m} \AAA)~,
}
where $\AAA$ and $\BBB$ are defined in \eqref{AB}, 
$k=1,\ldots n$ and $\tilde \phi_j: \mathbb{R}\to \mathbb{C}$,
$j=1,\ldots m$ are arbitrary Schwartz functions on $\mathbb{R}$,
denoted by $\hat\phi_j$
in \cite{Ho:1998hqa}.
The module of sections ${\cal  E}^\theta_{n,m}$
can be directly described in terms of these functions 
$\tilde \phi_j$,
that is, in terms of Schwartz functions  $\tilde
\phi:\mathbb{R}\times \mathbb{Z}_m\to \mathbb{C}$, 
thus recovering the more mathematical presentation used in
\cite{Connes:1997cr} (see also \cite{Schwarz:1998qj, Konechny:2000dp}) of the module ${\cal
  E}^\theta_{n,m}\in \M_{T_{(-\theta)}}$.
We denote  by $\triangleleft$, the action of the torus algebra 
on  the sections $\tilde \phi$. The action of the generators $U_1=e^{i\sigma^1},
U_2=e^{i\sigma^2}$  is induced by that on $\Tphi$ by defining
$\tilde\phi\triangleleft  U_\mu$ ($\mu=1,2$) such that
\eq{(\Tphi_k\star U_\mu)(\sigma^1,\sigma^2)=
\,\underset{s\in\mathbb{Z}}{\sum}\overset{m}{\underset{j=1}{\sum}}\,E(\AAA,\BBB)\star
(\tilde \phi\triangleleft  U_\mu )^{}_j (\tfrac{n}{m} \AAA)\;.
}
This gives  (use $U_1^{-1}\star \tilde\phi_j(\frac{n}{m}A)\star U_1=\tilde\phi_j(\frac{n}{m}A+\theta)$ and \eqref{EAlaB})
\eq{\label{phijtilde}
  (\tilde \phi\triangleleft U_1 )^{}_j(x)&=\; \tilde
  \phi_{j-1}(x-\tfrac{n}{m} + \theta)\;,
    \\[.2em]
  (\tilde \phi \triangleleft U_2)^{}_j(x)&=  \tilde\phi_j(x)\,e^{2\pi i (x-jn/m)}\;.
}
Thus we have a module isomorphism  between the module of sections
$\Tphi_k(\sigma^1,\sigma^2)$ (satisfying the boundary conditions) and
that of sections $\tilde \phi_j(x)$; we identify these two modules
over the noncommutative torus and use the same notation ${\cal
  E}^\theta_{n,m}$. 
Equations \eqref{phijtilde} provide a more explicit definition of the
module of sections  ${\cal  E}^\theta_{n,m}$ on the noncommutative
torus $T_{(-\theta)}$ because the twisted periodicity
conditions \eqref{qperiodicity} have been solved, it is however less
geometric than the implicit one with the constrained sections
${\Tphi}$. 
The geometric description holds for $n\in \mathbb{N}-\{0\}$ and $m\in
\mathbb{Z}-\{0\}$; then we have the modules ${\cal  E}^\theta_{n,0}$,
that are defined to be the direct sums
${T_{(-\theta)}}^{\oplus n}$ of $n$ copies of
the trivial module ${T_{(-\theta)}}$, i.e.,  the modules of sections of the
trivial $U(n)$-bundles on the noncommutative torus.
The algebraic definition \eqref{phijtilde} allows to consider also the
case $-n\in \mathbb{N}-\{0\}$ but there is no new module since 
${\cal  E}^\theta_{n,m}= {\cal E}^\theta_{-n,-m}$. Finally
\eqref{phijtilde} defines also the modules ${\cal{E}}^\theta_{0,m}$ that however
coincide with the modules  ${\cal{E}}^{\theta+1}_{m,m}$. 
In particular ${\cal{E}}^\theta_{0,1}={\cal{E}}^{\theta+1}_{1,1}$ is the
module of sections of the
$U(1)$-bundle over ${T_{(-\theta-1)}}=T_{(-\theta)}$ with charge $m=1$. It will play
a key role in \S 5.
\sk

Similarly to the classical case, covariant derivatives $D_\mu$ on the module of
sections ${\cal E}^\theta_{n,m}$ are described by 
covariant derivatives of the trivial $U(n)$-bundle on $\mathbb{R}^2_\theta$
satisfying the appropriate periodicity:
\eq{\label{qconnperiodicity}
  D_\mu\vert_{(\sigma^1+2\pi,\sigma^2)} &=\,\Omega_1(\sigma^2)\star D_\mu\vert_{(\sigma^1,\sigma^2)}\Omega_1^\dagger(\sigma^2)\;,\\[.2em]
  D_\mu\vert_{(\sigma^1,\sigma^2+2\pi)} &=\,\Omega_2(\sigma^1)\star D_\mu\vert_{(\sigma^1,\sigma^2)}\Omega_2^\dagger(\sigma^2)\;.
}
A particular solution to these conditions, that in the commutative
limit reduces to the previous solution, is
\eq{\label{CDNC}
D_1= \partial_{\sigma^1}-iA^\theta_1=\partial_{\sigma^1}~~, ~~~D_2=\partial_{\sigma^2}-iA^\theta_2 = \partial_{\sigma^2} -
\frac{i}{2\pi}\frac{m\sigma^1}{n-m\theta}{\mathbf{1}}~.}
Notice that the derivations $\partial_{\sigma^\mu}$ on
$T_{(-\theta)}$ can be defined intrinsically by $\partial_{\sigma^\mu}U_\nu=i\delta_{\mu\nu}U_\nu$,
and therefore independently from the realization
  $U_\mu=e^{i\sigma^\mu}$ with $[\sigma^\mu,\sigma^\nu]_\star=2\pi
  i\theta^{\mu\nu}$.

The curvature corresponding to the connection $D_\mu$ is
$F^\theta=i[D_1,D_2]_{\star}=\frac{1}{2\pi}\frac{m}{n-m\theta}{\mathbf{1}}$, a constant 
proportional to the unit matrix.
One can check that, taking into account the appropriate normalization \cite{Connes:1997cr}
of the integral on the noncommutative torus, the topological charge is indeed 
 $\frac{1}{2\pi} \int tr(F^\theta)=m$.
Since the covariant derivatives satisfy the Heisenberg algebra the
modules ${\cal E}^\theta_{n,m}$ with connection $(A^\theta_1,A^\theta_2)$ as in \eqref{CDNC}
are called Heisenberg modules. 
 \\

The algebra of endomorphisms of ${E^\theta}:=(\mathbb{R}^2_\theta)^{\oplus n}$ is
${\textrm{End}}(E^\theta)= (\mathbb{R}^2_\theta)^{\oplus (n\times n)}$,
that is, that of all $n\times n$ matrix valued noncommutative
functions on $\mathbb{R}^2$. The action on sections is the matrix
transformation  $\phi^\theta\mapsto \Psi^\theta\star\phi^\theta$. 
Of course, since this action is via $\star$-multiplication form the left, it  commutes with the 
action from the right of $\mathbb{R}^2_\theta$, we therefore have the
bimodule $E^\theta\in {}_{\textrm{End}(E^\theta)}\M_{\mathbb{R}^2_\theta}$.
The algebra of endomorphisms
$\textrm{End}({\cal{E}}^\theta_{n,m})$ is the subalgebra of $\textrm{End}(E^\theta)$
that preserves the twisted boundary conditions \eqref{classbd}: if
$\phi^\theta$ satisfies \eqref
{classbd} then so does $\Psi^\theta\star \phi^\theta$. That is, endomorphisms of ${\cal
  E}^\theta_{n,m} $ are endomorphisms of $E^\theta$ that satisfy the
twisted boundary conditions in the adjoint representation
\eq{\label{NCendperiodicity}
  \Psi^\theta(\sigma^1+2\pi,\sigma^2) &=\,\Omega_1(\sigma^2)\star
  \Psi^\theta(\sigma^1,\sigma^2)\star \Omega_1^{-1}(\sigma^2)\;,\\[.2em]
  \Psi^\theta(\sigma^1,\sigma^2+2\pi)&=\,\Omega_2(\sigma^1)\star
  \Psi^\theta(\sigma^1,\sigma^2)\star \Omega_2^{-1}(\sigma^2)\;.
}
The algebra $\textrm{End}({\cal{E}}^\theta_{n,m})$ is generated by
$\star$-multiplication with the $U(n)$-valued functions (see
\cite{Morariu:1998qm} for a proof)
\eq{\label{Ztheta}Z^\theta_1 = e^{\frac{i\sigma^1}{n - m\theta}}V^b~~,~~~Z^\theta_2 =
e^{\frac{i\sigma^2}{n}}U^{-b}~.}
Since $Z_1^\theta\star Z_2^\theta=e^{2\pi i \check\theta}Z^\theta_2\star Z^\theta_1$,
with  $\check\theta=\frac{a(-\theta)+b}{m(-\theta)+n}$, and  $an-bm=1$,
$a,b\in \mathbb{Z}$, we see that with this choice of generators the endomorphisms algebra is the torus algebra
$T_{\check\theta}$.
Thus we have the 
bimodule ${\cal{E}}^\theta_{n,m}\in {}^{}_{T_{\check\theta}}\M^{}_{T_{(-\theta)}}$.
\sk
The connection  \eqref{CovDer} and the endomorphisms can be described directly on the
solutions $\tilde\phi$ rather than on the more geometric sections $\Tphi$. 
Proceeding as before, by star-multiplying from the left with the
$U(n)$-valued functions
$Z^\theta_1 = e^{\frac{i\sigma^1}{n - m\theta}}V^b$ and $Z^\theta_2 =
e^{\frac{i\sigma^2}{n}}U^{-b}$, as well as with the covariant derivatives
$D_1= \partial_{1}, D_2 = \partial_{2} -
\tfrac{i}{2\pi}\tfrac{m\sigma^1}{n-m\theta}\mathbf{1}$, we arrive at the
following proposition, cf. \cite{Ho:1998hqa}, 

\begin{prop}
  \label{inducelemma}
Star multiplication from the left on the sections
$\Tphi(\sigma^1,\sigma^2)$ with the functions $Z_\mu$ induces the
action $Z_\mu\,\triangleright \,\tilde\phi$ of $Z_\mu$ on the
Schwartz functions  $\tilde\phi: \mathbb{R}\times \mathbb{Z}_m\to \mathbb{C}$. Similarly, acting with covariant derivatives $D_\mu$ on
$\Hop(\sigma^1,\sigma^2)$ induces the action 
$D_{\mu\,} \tilde\phi$. Explicitly, using the shorthand notation (\ref{shand}),
\eq{
  (Z^\theta_\mu\star \Tphi)^{}_k\,(\sigma^1,\sigma^2)&=
\,\underset{s\in\mathbb{Z}}{\sum}\overset{m}{\underset{j=1}{\sum}}\,E(\AAA,\BBB)\star 
(Z^\theta_\mu\triangleright \tilde \phi) ^{}_j (\tfrac{n}{m} \AAA)~,~~\\
(D_\mu\, \Tphi) ^{}_k\,(\sigma^1,\sigma^2)&=
\,\underset{s\in\mathbb{Z}}{\sum}\overset{m}{\underset{j=1}{\sum}}\,E(\AAA,\BBB)\star 
(D_{\mu\,}\tilde \phi)^{}_j (\tfrac{n}{m} \AAA)~,
}
with
  \begin{align}
    (Z^\theta_1 \triangleright \tilde \phi)^{}_j(x) &=\; \tilde \phi_{j-a}(x-\tfrac{1}{m}) \;,\label{D0}\\
    (Z^\theta_2 \triangleright \tilde \phi)^{}_j(x) &=\; \tilde \phi_j(x)\,e^{2\pi i (\tfrac{x}{n-m\theta} -\tfrac{j}{m})}\;,\\
    (D_1  \,\tilde \phi)^{}_j(x) &=\; \frac{imx}{n-m\theta}\,\tilde \phi_j(x)\;,\label{D1}\\
    (D_2  \,\tilde \phi)^{}_j(x)&=\;
                                           \frac{1}{2\pi}\,\frac{\partial}{\partial
                                           x}\tilde \phi_j(x)\;
                                           ,\label{D2}
  \end{align}
  where in \eqref{D0}, $a \in \mathbb{Z}$ is such that $an-bm = 1$
  with $b\in  \mathbb{Z}$.
\end{prop}
\sk

For later use we observe that the connection $D_\mu$
induces derivations $\hat\delta_1$, $\hat\delta_2$, of the algebra of endomorphisms of
${\cal{E}}_{n,m}^{\theta} $ via the definition
$\hat\delta_\mu(\Psi^{\theta})\triangleright\phi^{\theta}:=
D_\mu(\Psi^\theta \triangleright\phi^{\theta})-\Psi^{\theta}\triangleright D_\mu\phi^{\theta}$.
Since $D_\mu$ is a constant curvature connection it is easy to compute
$[\hat\delta_1,\hat\delta_2]=0
$, so that
$\hat\delta: \partial_{\sigma^\mu}\mapsto \hat\delta_\mu$ is a
representation on $\textrm{End}({\cal{E}}_{n,m}^{\theta})$ of the
abelian Lie algebra $\mathbb{R}^2$ spanned by the derivations
$\partial_{\sigma^1}, \partial_{\sigma^2}$ on the torus
$T_{{(-\theta})}$. A simple explicit calculation, setting  $\Psi^{\theta}=
Z_\nu^{\theta} $ and using \eqref{D0}-\eqref{D2} gives
\eq{\hat\delta_\mu(Z_\nu^{\theta})=\frac{i}{n-m\theta}  \delta_{\mu\nu}
Z_\nu^{\theta}~.\label{inducedderiv} 
}
On the other hand
$\textrm{End}({\cal{E}}_{n,m}^{\theta})\simeq T_{\check\theta}$ is
generated by $Z_1^{\theta}, Z_2^{\theta}$ and has canonical
derivations $\partial_{\check\sigma^\mu}$ defined by
$\partial_{\check\sigma^\mu}(Z_\mu^{\theta})=i\delta_{\mu\nu}Z^\theta_\nu$ (the
notation used follows from writing $Z_\mu^{\theta}=e^{i\check\sigma^\mu}$ with
$[\check\sigma^\mu,\check\sigma^\nu]_{\check\star}=-2\pi i \check \theta^{\mu\nu}$).
We thus conclude that
\eq{\hat\delta_\mu=\frac{1}{n-m\theta}\partial_{\check\sigma^\mu}~.\label{CME}}
By definition a complete (or gauge) Morita equivalence bimodule 
 ${\cal E} \in {}_{T_{\check\theta}}\M^{}_{T(-\theta)}$ is a bimodule
 with a (right module) constant curvature connection $D_\mu$ proportional to the
identity and such that the induced derivations $\hat\delta_\mu$ on
$\textrm{End}({\cal{E}}^\theta_{n,m})\simeq T_{\check\theta}$  are an invertible
linear combination of the canonical derivations
$\partial_{\check\sigma^\mu}$ on $T_{\check\theta}$, see \cite{Schwarz:1998qj}. We have seen that
the Heisenberg bimodules ${\cal{E}}_{n,m}^{\theta} $ are complete
Morita equivalence bimodules.

Notice that $D^L_\mu:=(n-m\theta)D_\mu$ satisfies the left Leibnitz
rule 
\eq{\label{DLconn}D^L_\mu(Z^\theta_\nu\triangleright
\tilde\phi)=\partial_{\check\sigma^\mu}
(Z^\theta_\nu)\triangleright\tilde\phi+Z^\theta_\nu\triangleright D^L_\mu\tilde\phi=
i\delta_{\mu\nu}Z^\theta_{\nu}\triangleright\tilde\phi+Z^\theta_\nu\triangleright
D^L_\mu\tilde\phi}
hence $D^L_\mu$ is a left connection on the left
$T_{\check\theta}$-module ${\cal{E}}_{n,m}^{\theta}$. Another way of
characterizing complete (or gauge)  Morita equivalence bimodules ${\cal E} \in {}_{T_{\check\theta}}\M^{}_{T(-\theta)}$ is then by
requiring the right $T_{(-\theta)}$-module constant curvature connection $D_\mu$ to be
also, up to an invertible linear transformation, a left ${T_{\check\theta}}$-module connection. 
\sk
In this section we have described noncommutative vector bundles
using $\star$-products. The advantage of this deformation quantization
approach 
is that  we have a
manifest dependence on the deformation parameter $\theta$. This
naturally leads to generalize the Seiberg--Witten map of Section \ref{SWl}
by establishing a Seiberg--Witten map between  classical sections 
in ${\cal E}_{n,m}$ and quantum sections
in $ {\cal  E}^\theta_{n,m}$.
This Seiberg--Witten map is compatible with the bimodule
structure of the Heisenberg modules
\eq{{\cal{E}}^\theta_{n,m}\in
{}_{T_{\check\theta}}\M^{}_{T_{(-\theta)}}~~,~~~ \check\theta=\frac{a(-\theta)+b}{m(-\theta)+n}}
with connection $A^\theta_\mu$ and constant curvature
$F^\theta=\frac{1}{2\pi}\tfrac{m}{n-m\theta}\mathbf{1}$.

\section{The Seiberg--Witten map on tori}

In the preceding section we have described $U(n)$-vector bundles with
connections and topological charge $m$ on tori and on
noncommutative tori. On the other hand in the first section we have
recalled that the Seiberg--Witten map relates commutative to
noncommutative gauge theories. 
Here we first see  that it is a
quantization of
$U(n)$-bundles with connections on $\mathbb{R}^2$ to 
$U(n)$-bundles with connections on noncommutative
$\mathbb{R}^2$.
Then we construct an induced Seiberg--Witten map that quantizes 
$U(n)$-bundles with connections on $\TT2$
to $U(n)$-bundles with connections on $T_\theta$. 
While the treatment in Section \ref{SWl} and \ref{SWUNR} is local, because on (one
open chart)
$\mathbb{R}^2$, 
in Section \ref{SWNCT} we achieve a global description of the  Seiberg--Witten map on
tori. In Section \ref{SWHO} we then compare the general construction we perform with the
description in Section \ref{UNNC} of bundles on noncommutative  tori in terms of the module of noncommutative sections 
$ {\cal  E}^\theta_{n,m}$. We find full ageement. On the one hand,
this frames the solution found in $ {\cal  E}^\theta_{n,m}$ in the general
Seiberg--Witten map deformation scheme. On the other hand, it
provides an explicit solution to the Seiberg--Witten equations and
an example of a  formal deformation quantization that is actually non-formal, since $\theta$ can be
specialized to a real number and power series in $\theta$
can be summed and analytically continued to  well defined complex
valued functions.

\subsection{Seiberg--Witten map for bundles on $\mathbb{R}^2$}\label{SWUNR}
The Seiberg--Witten map presented in Section \ref{SWl} can be
seen as a quantization of $U(n)$-bundles with connections on $\mathbb{R}^2$.
Let's describe a bundle  with connection, togheter with a Poisson
structure $\theta$ on $\mathbb{R}^2$ via the triple
$(\,E\in\prescript{}{\textrm{End}(E)}{\M}_{\CCi(\mathbb{R}^2)}, A_\mu, \theta)$, where 
$E=C^\infty({\mathbb{R}^2})^{\,\oplus n}$ is the module of sections of
the trivial bundle $\mathbb{R}^2\times\mathbb{C}^n\to \mathbb{R}^2$,
it is a bimodule $E\in
\prescript{}{\textrm{End}(E)}\M_{\CCi(\mathbb{R}^2)}$, and $A_\mu~
(\mu=1,2)$ are the components of a connection.
The Seiberg--Witten map provides a quantization of this bundle
to a $U(n)$-bundle with connection on the noncommutative  plane
$\mathbb{R}^2_\theta$:
\eq{\big(E\in\prescript{}{\textrm{End}(E)}{\M}_{\CCi(\mathbb{R}^2)}, A_\mu, \theta\big) 
    \xrightarrow{~\textrm{SW map}_{}~} 
    \big(\hat E=E^\theta\in\prescript{}{\textrm{End}(E^\theta)}{\M}_{\mathbb{R}^2_\theta},
    \hat{A}_\mu\big)
\,,\label{SWR2}}
where the $\hat{A}_\mu$ are determined by the recursive relation
\eqref{An+1} of the Seiberg--Witten map for connections.
The noncommutative connection is a connection on the module of sections
${\hat E}=E^\theta:= ({\mathbb{R}^2_\theta})^{\oplus n}$ (the trivial $U(n)$-vector
bundle on $\mathbb{R}^2_\theta$). To see this, we
observe that sections
$\phi\in E=C^\infty(\mathbb{R}^2)^{\oplus
  n}$ are mapped
via Seiberg--Witten map to sections $\hat\phi\in 
  ({\mathbb{R}^2_\theta})^{\oplus n}$ (the former transform under
usual gauge transformations the latter under noncommutative gauge
transformations) and that the covariant derivative
$D_\mu=\partial_\mu-iA_\mu: E\to E$  is mapped to the covariant derivative
$\partial_\mu-i\hat A_\mu:\hat E\to \hat E$ (that acts via
$\star$-multiplication by $\hat A_\mu$).

Similarly, bundle endomorphisms $\Psi\in \textrm{End}(E)=C^\infty(\mathbb{R}^2)^{\oplus (n\times n)}$,
i.e., sections of the trivial bundle
$\mathbb{R}^2\times\mathbb{C}^{n\times n}\to \mathbb{R}^2$, that
transform in the adjoint representation, are mapped via Seiberg--Witten map to bundle endomorphisms
$\hat\Psi\in \textrm{End}(\hat E)=({\mathbb{R}^2_\theta})^{\oplus (n\times n)}$, i.e.,
sections transforming in the adjoint representation (with $\star$-product multiplication).

\sk
Let's explicitly compute the Seiberg--Witten map for the trivial
bundle on $\mathbb{R}^2$ with connection $(A_1=0,
A_2=\frac{m\sigma^1}{n 2\pi}\mathbf{1})$ as in \eqref{CovDer}.
We see that since $A_1=0$, and $A_2$ depends only on the variable
$\sigma^1$, then $A_1^k=0$, and $A_2^k$ depends only on $\sigma^1$, so
that $\hat{A}_1=0$, and $\hat{A}_2$ depends only on $\sigma^1$. It
follows that $\partial_1\hat{A}_2=\hat{F}_{12}$, henceforth
\eqref{An+1} simplifies to 
\eq{
A^{k+1}_2=\frac{\pi\,\theta}{n+1}\,\frac{\partial}{\partial \sigma^1}\sum_{p=0}^{k}
A^p_2\,{A}^{k-p}_2~.}
It is then easy to solve also for $\hat{A}_2$, and to obtain the
noncommutative connection 
\eq{
  \label{connectionfixed}
  \SWA_1 &=\; 0 \;, \\
  \SWA_2 &=\;
A_2\,\sum_{k=0}^\infty\,\Bigl(\frac{m}{n}\theta\Bigr)^k=\; \frac{1}{2\pi}\frac{m\,\sigma^1}{n - m\theta}\mathbf{1}\;.}
Notice that  $(\SWA_1,\SWA_2)=(A^\theta_1,A^\theta_2)$ as defined
in \eqref{CDNC}.

We can also compute the Seiberg--Witten map for the endomorphisms
$\Omega_\al\in \textrm{End}(E) =C^\infty(\mathbb{R}^2)^{\oplus
  (n\times n)}$ defined in \eqref{Om12}. From the recursive solution
for sections in the adjoint \eqref{phiadrec}, recalling that $\hat A_1=0$, it
is easy to see that the noncommutative endomorphisms $\hat\Omega_\al\in \textrm{End}(\hat E) =({\mathbb{R}^2_\theta})^{\oplus (n\times n)}$
coincide   with the commutative ones:
\eq{\label{Omfixed}\hat\Omega_1=\Omega_1=e^{im\sigma^2/n}U~,~~~\hat\Omega_2=\Omega_2= V~.}
Notice that we obtain this same result if instead of the
Seiberg-Witten map \eqref{Om12} we use the more general one
\eqref{extendedSWPA} as long as $\hat C_{12}(\Omega_\al, \hat A)=0$.
Notice also that by choosing a specific expression for the gauge
potential we have fixed the gauge and therefore do not need to consider
the Seiberg-Witten maps \eqref{rtwo} or \eqref{extendedSWE}  quantizing local infinitesimal gauge transformations.

\subsection{Seiberg--Witten map for bundles on $\TT2$
}\label{SWNCT}
Since $\TT2\simeq \mathbb{R}^2/(2\pi \mathbb{Z})^2$ we have
$C^\infty(\TT2)\hookrightarrow  C^\infty(\mathbb{R}^2)$ as the subalgebra of
$2\pi$-periodic functions, moreover, sections of bundles on $\TT2$ can be seen as
sections of bundles on $\mathbb{R}^2$ satisfying twisted
periodicity conditions. Similarly 
$\TT2_{(-\theta)}\hookrightarrow\mathbb{R}^2_\theta$, and 
sections of modules on $\TT2_{(-\theta)}$ can be seen as
sections of the module $\hat E=(\mathbb{R}^2_\theta)^{\oplus n}$ on
$\mathbb{R}^2_\theta$ satisfying noncommutative twisted periodicity conditions. 

Let $(\mathcal{E},  A_\mu)$ be the module of sections and the
connection of a bundle on a torus determined by: 
{\it i)} a trivial bundle 
$(E=C^\infty({\mathbb{R}^2})^{\,\oplus n},  A_\mu)$ on the plane
and {\it ii)} $U(n)$-valued
functions on the plane $\Omega_1(\sigma^2)$, $\Omega_2(\sigma^1)$
satisfying the classical cocycle condition  \eqref{classcy}. Let $\hat A_\mu$, $\hat\Omega_1(\sigma^2)$, $\hat\Omega_2(\sigma^1)$
be the corresponding Seiberg-Witten map quantizations. Let these latter
satisfy the noncommutative cocycle conditions 
\eq{\hat\Omega_1(\sigma^2 + 2\pi)\star\hat \Omega_2(\sigma^1) =
  \;\hat\Omega_2(\sigma^1 + 2\pi)\star \hat\Omega_1(\sigma^2)\;, }
and, toghether with  $\hat A_\mu$, the analogue of the noncommutative twised periodicity
conditions \eqref{qconnperiodicity}. Then,

\begin{definition} We denote by
$(\hat{\mathcal{E}}, \hat A_\mu)$  the Seiberg--Witten map quantization
of $(\mathcal{E},  A_\mu)$, where  
$\hat\CE$ is the subset
of all elements in $\hat E= ({\mathbb{R}^2_\theta})^{\oplus n} $ that
satisfy the noncommutative twisted periodicity
conditions  \eqref{qperiodicity} with $\hat\Omega_1(\sigma^2)$,
$\hat\Omega_2(\sigma^1)$.
\end{definition}
We also write  $({\mathcal E},A_\mu)\stackrel{\rm{SW\,
  induced}}{-\!\!\!-\!\!\!-\!\!\!-\!\!\!\longrightarrow}(\hat {\mathcal E},\hat A_\mu)$
because the quantization of this bundle on the torus is induced by the
usual Seiberg--Witten map on the plane.
It is immediate to see that $\hat\CE$  is a right
$T_{(-\theta)}$-module.

An example is given by $(\hat{\mathcal{E}}_{n,m}, \hat A_\mu )$,  the Seiberg--Witten quantization of
$(\mathcal{E}_{n,m},  A_\mu)$,  where $(A_\mu)=(A_1,A_2)=(0,
\frac{m\sigma^1}{n 2\pi}\mathbf{1})$. It is defined by the quantized connection and
endomorphisms $\hat A_\mu$ and $\hat\Omega_\al$ computed in
\eqref{connectionfixed} and \eqref{Omfixed}. They coincide with the
connection of $\mathcal{E}^\theta_{n,m}$ and the endomorphisms defining $\mathcal{E}^\theta_{n,m}$,
cf. \eqref{CDNC} and \eqref{Omperiodicity}. This shows
$(\hat{\mathcal{E}}_{n,m}, \hat A_\mu )=(\mathcal{E}^\theta_{n,m}, A^\theta_\mu )$.

We sharpen this result with 
 the following commutative diagram:
\begin{equation}
      \label{scheme}
  \begin{matrix}
    \big(E\in\prescript{}{\textrm{End}(E)}{\M}_{\CCi(\mathbb{R}^2)}, A_\mu,
    \theta\big) &
    \!\xrightarrow{~~\textrm{SW map}_{}~~} \!&
    \big(\hat E=E^\theta\in\prescript{}{\textrm{End}(E^\theta)}{\M}_{\mathbb{R}^2_\theta},
    \hat{A}_\mu\big)\\
    {\scriptstyle{i}}\!\!{\longhookuparrow } &  & {\scriptstyle{i_\theta}} \!\!{\longhookuparrow } \\
   \big( {\mathcal{E}}_{n,m}\in\prescript{}{\textrm{End}({\mathcal{E}}_{n,m})}{\M}_{\CCi(\TT2)}, A_\mu,\theta\big)
    &\! \xrightarrow{~~\textrm{SW induced}_{}~~}\! &
    \big(\hat{\mathcal{E}}_{n,m}={\mathcal{E}}_{n,m}^\theta\in\prescript{}{\textrm{End}({\mathcal E}^\theta_{n,m})}{\M}^{}_{\TT2_{(-\theta)}}, \hat{A}_\mu\big)
  \end{matrix}
\end{equation}
\\
where by
$\mathcal{E}_{n,m}\in{}_{\textrm{End}({\mathcal{E}}_{n,m})} {\M}_{\CCi(\TT2)}\stackrel{i}{\hookrightarrow}E\in{}_\textrm{{End}(E)}{\M}_{\CCi(\mathbb{R}^2)}$ 
we mean that the module of sections ${\mathcal{E}}_{n,m}$ is a linear
subspace of
${{E}}$,\footnote{With slight
  abuse of teminology, we call $\mathbb{C}[[\theta]]$-modules (and
  $\mathbb{C}[[\theta]]$-submodules) simply linear spaces (and
  subspaces).} 
and that it is a bimodule over the
subalgebras $\CCi(\TT2) \hookrightarrow\CCi(\mathbb{R}^2)$ and
$\textrm{End}(\mathcal{E}_{n,m}) \hookrightarrow\textrm{End}({E})$,
and similarly for the other  map ${i_\theta}$.

\sk
\begin{theorem}\label{T:SWT} The induced Seiberg--Witten map 
on torus bundles 
$({\mathcal E}_{n,m},A_\mu)\stackrel{\rm{SW} \,\rm{
  induced}}{-\!\!\!-\!\!\!-\!\!\!-\!\!\!\longrightarrow}(\hat {\mathcal E}_{n,m},\hat A_\mu)$
satisfies the commutative diagram \eqref{scheme}: \\
i) Let $\phi\in E$
satisfy the  twisted boundary conditions \eqref{classbd}, then its Seiberg--Witten
  quantization $\hat\phi$ satisfies the twisted boundary conditions
  \eqref{qperiodicity}, hence $\hat\phi\in {\mathcal E}_{n,m}^\theta$\\
ii) Let $\Psi\in  \textrm{End}({E})$
  satisfy the twisted boundary conditions \eqref{endperiodicity},
then its Seiberg--Witten
  quantization $\hat\Psi$ satisfies the twisted boundary conditions
\eqref{NCendperiodicity}, hence $\hat\Psi\in  {\textrm{End}}({\mathcal E}_{n,m}^\theta)$.

\end{theorem}
\begin{proof}
We have already shown that  $A_\mu\to \hat A_\mu=\hat A_\mu^\theta$,
$\Omega_\al\to \hat\Omega_\al=\Omega_\al$, proving that 
$({\mathcal{E}}_{n,m},  A_\mu )\rightarrow  (\hat{\mathcal{E}}_{n,m}, \hat A_\mu )=(\mathcal{E}^\theta_{n,m}, A^\theta_\mu )$.  

{\it i)} Let $\SWP|_{\sigma+2\pi}=(\Omega\star \SWP)|_\sigma$ be a shorthand
notation for $\SWP(\sigma^1+2\pi,\sigma^2)=\Omega_1(\sigma^2)\star
\SWP(\sigma^1,\sigma^2)$ as well as for  
$\SWP(\sigma^1,\sigma^2+2\pi)=\Omega_2(\sigma^1)\star
\SWP(\sigma^1,\sigma^2)$.
We show that 
\eq{
\SWP|_{\sigma+2\pi}=(\Omega\star \SWP)|_\sigma~~\Rightarrow
~~
\SWPp|_{\sigma+2\pi}=(\Omega\star'\SWPp)|_\sigma~,
\label{tt'}}
where we recall that $\star'$ and $\SWPp$ are the star product and the
Seiberg--Witten map at noncommutativity parameter $\theta'$. The proof
then follows by setting $\theta=0$ and $\theta'=\theta$ in \eqref{tt'}.
Recalling the uniqueness of the recursive solution \eqref{phirec} for
$\SWP$, in order to prove \eqref{tt'}, it is enough to prove its infinitesimal
version for $\theta'=\theta+\delta\theta$, i.e., 
\eq{
\delta^\theta\SWP_{}\vert_{\sigma+2\pi}=\delta^\theta(\Omega\star\SWP)\vert_{\sigma}~.
\label{dpdop}}
The right hand side reads
$$
\delta^\theta(\Omega\star\SWP)\vert_{\sigma}=(\Omega\star\delta^\theta\SWP)\vert_{\sigma}+i\pi\delta\theta^{\mu\nu}(\partial_\mu\Omega\star\partial_\nu\SWP)\vert_{\sigma}=-\pi\delta\theta^{21}(\Omega\star
\SWA_2\star D_1\SWP-
  i\partial_2\Omega\star D_1\SWP)\vert_{\sigma}
$$
where we used that $\Omega_1$ and $\Omega_2$ are independent of
$\theta=\theta^{12}=-\theta^{21}$ (cf. \eqref{Omfixed}),  the
Seiberg--Witten differential equation \eqref{matterfund} for
$\SWP$ with $\SWA_1=0$ (cf. \eqref{connectionfixed}) so that
$\partial_1\SWP=D_1\SWP$, and we also used that
$\partial_1\Omega_1=\partial_1\Omega_2=0$.
We proceed similarly with the left hand side and obtain 
\begin{eqnarray}
\delta^\theta\SWP_{}\vert_{\sigma+2\pi}
&=&-\pi\delta \theta^{21}(\SWA_2\star D_1\SWP)\vert_{\sigma+2\pi}
=-\pi\delta \theta^{21}\SWA_2\vert_{\sigma+2\pi}\star
D_1\SWP\vert_{\sigma+2\pi}\nonumber \\[.4em]
&=&-\pi\delta \theta^{21}(\Omega\star\SWA_2\star
D_1\SWP+i\Omega\star \partial_2\Omega^{-1}\star\Omega\star D_1\SWP)
    \nonumber\\[.4em]
&=&\delta^\theta(\Omega\star\SWP)\vert_\sigma~
\end{eqnarray}
where in the first line we used invariance of the $\star$-product under constant translations,
while in the second line the twisted periodicity conditions satisfied
by  the
connection: 
$\SWA_2\vert_{\sigma+2\pi}=(\Omega\star\SWA_2\star
\Omega^{-1}+i\:\!\Omega\star\partial_2\Omega^{-1})\vert_\sigma$, cf. \eqref{qconnperiodicity},
and by the section $D_1\SWP$.
This proves {\it i)}. The proof of  {\it ii)} is very similar and left to the reader.
\end{proof}

More in general, as it is clear from the proof, the  induced
Seiberg--Witten map quantizes $U(n)$-bundles on noncommutative tori:
\eq{  \label{SWTT'}
\big(\hat{\mathcal{E}}_{n,m}={\mathcal{E}}_{n,m}^\theta\in\prescript{}{\textrm{End}({\mathcal E}^\theta_{n,m})}{\M}^{}_{\TT2_{(-\theta)}}, \hat{A}_\mu\big)
     \xrightarrow{~~\textrm{$SW_\theta^{\theta'}$}_{}~~} 
    \big(\hat{\mathcal{E}}^{{\!\!\:\!\!\!\phantom{I}}^{{'}}}_{n,m}={\mathcal{E}}_{n,m}^{\theta'}\in\prescript{}{\textrm{End}({\mathcal
        E}^{\theta'}_{n,m})}{\M}^{}_{\TT2_{(-\theta')}}, \hat{A}^{{\!\!\:\!\!\!\phantom{I}}^{{'}}}_\mu\big)\,.
}

We also notice that there are two ways to quantize the algebra of endomorphisms
$\textrm{End}({\mathcal E}_{n,m})$: on the one hand $\textrm{End}({\mathcal
  E}_{n,m})$ is just the algebra $T_{{\tilde\theta}|_{\theta=0}}$ with
$\tilde\theta\vert^{}_{\theta=0}=\frac{a(-\theta)+b}{m(-\theta)+n}\vert^{}_{\theta=0}$, and
its quantization with respect to $0\to\theta$ is $\widehat{\textrm{End}({\mathcal E}_{n,m})}=T_{\tilde{\theta\,}}$
with $\tilde{\theta\,}=\frac{a(-\theta)+b}{m(-\theta)+n}$.
On the other hand we can first quantize ${\mathcal E}_{n,m}$ to 
$\hat{\mathcal{E}}_{n,m}$
 and then consider the algebra of
endomorphisms of this latter: 
$\textrm{End}(\hat{\mathcal{E}}_{n,m})$.
Since
$\hat{\mathcal{E}}_{n,m}={\mathcal{E}}_{n,m}^{\theta}$
as a corollary of Theorem \ref{T:SWT} we immediately have that these two alternative
quantization routes are equivalent:
$\label{hatEcomp}
\widehat{\textrm{End}({\mathcal E}_{n,m})}=
\textrm{End}(\hat{\mathcal{E}}_{n,m})~.
$
\sk
In Section \ref{sec:ambig} we have studied the ambiguities in the
Seiberg--Witten differential equations and seen that they lead in
general to different Seiberg--Witten maps. In the present context and
for later use  we focus on the ambiguities in the Seiberg--Witten map
arising from allowing a nonvanishing term $\hat C_{\mu\nu}$ for matter
fields (and keeping $\hat D_{\mu\nu k}=0$ and $\hat E_{\mu\nu}=0$). We can generalize
Theorem \ref{T:SWT} to this case 

\begin{theorem}\label{TSWText}
Consider the generalized Seiberg-Witten map defined by the
differential equations \eqref{extendedSWA}-\eqref{extendedSWPA} with
$\hat D_{\mu\nu}=0$, $\hat E_{\mu\nu} =0$ and arbitrary  
$\hat C_{\mu\nu}=\hat C_{\mu\nu}(D_\mu, \hat{F},\hat\phi,\theta)$
transforming in the fundamental, respectively $\hat C_{\mu\nu}=\hat
C_{\mu\nu}(D_\mu, \hat F,\hat\Psi,\theta)$ transforming in the
adjoint.  If $\hat C_{\mu\nu}(D_\mu, \hat F,\hat\Omega_\al,\theta)=0$,
$(i=1,2)$, the induced Seiberg--Witten map on torus bundles determined by diagram  \eqref{scheme}  quantizes $({\mathcal E}_{n,m},A_\mu)$
to $(\hat {\mathcal E}_{n,m},\hat A_\mu)=({\mathcal E}_{n,m}^\theta,
A^\theta_\mu)$, and it is well defined, consistently quantizing $\phi\in {\mathcal E}_{n,m}$ to $\hat\phi\in {\mathcal E}^\theta_{n,m}$ 
and $\Psi\in \textrm{End}({\mathcal E}_{n,m})$ to $\hat\Psi\in \textrm{End}({\mathcal 
  E}^\theta_{n,m})$. 
\end{theorem}
\begin{proof}
Since $\hat C_{\mu\nu}(D_\mu, \hat F,\hat\Omega_\al,\theta)=0$, Only the Seiberg-Witten map for the matter fields $\phi$ and $\Psi$ has changed, so
we still have  $(\hat {\mathcal E}_{n,m},\hat A_\mu)=({\mathcal  E}_{n,m}^\theta, A^\theta_\mu)$.
If $\hat C_{\mu\nu}=\hat C_{\mu\nu}(D_\mu, \hat{F},\hat\phi,\theta)$ transforms in the fundamental it is a
section of $\hat E=E^\theta=(\mathbb{R}^2_\theta)^{\oplus n}$.
Since $\hat E_{\mu\nu} =0$, condition
\eqref{constrC} is satisfied  (and similarly for
$\hat C_{\mu\nu}=\hat C_{\mu\nu}(D_\mu, F^\theta,\hat\Psi,\theta)$  in the adjoint). Since $D_\mu,
{F}^\theta$ and $\hat\phi$ all satisfy the twisted periodicity conditions
it follows that $\hat
C_{\mu\nu}\vert_{\sigma+2\pi}=(\Omega\star \hat C_{\mu\nu})\vert_\sigma$ 
i.e., that  $\hat C_{\mu\nu}$ is a section of ${\mathcal E}^\theta_{n,m}$
(and similarly $\hat C_{\mu\nu}(D_\mu, F^\theta,\hat \Psi,\theta)\in
\textrm{End}({\mathcal E}^\theta_{n,m})$).   We now repeat the proof
of Theorem \ref{T:SWT} and observe that the left hand side in
\eqref{dpdop} 
has the extra term $-\frac{\pi}{2}\delta\theta^{\mu\nu}\hat C_{\mu\nu}\vert_{\sigma+2\pi}$ while 
the right hand side 
has the extra term $-\frac{\pi}{2}\delta\theta^{\mu\nu}(\Omega\star \hat C_{\mu\nu})\vert_\sigma$. These two extra terms are equal  because
we just observed that $\hat C_{\mu\nu}$ is a section of ${\mathcal E}^\theta_{n,m}$. A similar argument holds for the adjoint case $\hat C_{\mu\nu}=\hat C_{\mu\nu}(D_\mu,F^\theta_{\mu\nu},\hat\Psi,\theta)$.
\end{proof}

Examples of nontrivial $\hat C_{12}$ terms for fields in the adjoint that
vanish on $\Omega_\al\in \textrm{End}(E) =C^\infty(\mathbb{R}^2)^{\oplus
  (n\times n)}$ are $\hat
C_{\mu\nu}(D_\mu,F^\theta_{\mu\nu},\Omega_\al,\theta)=-iD_1D_2\Omega_\al=-D_2D_1\Omega_1=0$
and $-iD_1D_1\Omega_\al=0$
(indeed recall the definition of $\Omega_\al$ in \eqref{Om12}, use that
$[D_1,D_2]\Psi=-i[F,\Psi]=0$ since $F$ is constant and proportional to
the unit matrix, and that $D_1=\partial_{\sigma^1}$).
These $\hat C_{12}$ terms therefore satisfy $\hat C_{\mu\nu}(D_\mu,
\hat F,\hat\Omega_\al,\theta)=0$ since $\hat\Omega_\al=\Omega_\al$.

\subsection{Explicit solutions: Ho's sections $\Hop$, $Z^\theta_\mu$\label{SWHO}}
We here show that the Seiberg-Witten quantization of  the sections
$\phi\in{\cal E}_{n,m}$ in the fundamental, and of the sections $\Psi\in
\textrm{End}({\cal E}_{n,m})$ in the adjoint, generated by $Z_\mu$, gives the sections
$\phi^\theta\in{\cal E}^\theta_{n,m}$ and $\Psi^\theta\in
\textrm{End}({\cal E}^\theta_{n,m})$ that are generated by
$Z_\mu^\theta$ as described in Section \ref{UNNC}. This shows that the solutions
presented in \cite{Ho:1998hqa} fits into the framework of
Seiberg--Witten maps. It also provides explicit closed form solutions
to the Seiberg--Witten map equations for nontrivial bundles.

 We begin with the generators $Z_\mu\in
\textrm{End}({\cal E}_{n,m})$. 
  \begin{prop}\label{endolemma}
The Seiberg--Witten quantization according to Theorem \ref{T:SWT} of the generators
$Z_\mu\in\textrm{End}({\cal E}_{n,m})$ defined in \eqref{commgenerat} as
$Z_1 = e^{\frac{i\sigma^1}{n}}V^b$,  $Z_2=e^{\frac{i\sigma^2}{n}}U^{-b}$,
gives the generators $\hat Z_\mu=Z_\mu^\theta\in\textrm{End}({\cal
  E}^\theta_{n,m})$ defined in \eqref{Ztheta} as $Z^\theta_1 = e^{\frac{i\sigma^1}{n-m\theta}}V^b$,  $Z^\theta_2=e^{\frac{i\sigma^2}{n}}U^{-b}$.
    \end{prop}
    \begin{proof}
The Seiberg--Witten equation \eqref{matteradj} for sections in the adjoint
representation with connection \eqref{connectionfixed} reads
      \begin{equation}
        \label{endoeq}
        \frac{\partial \hat Z_\mu}{\partial \theta} =\;\pi\,\bigl\lbrace \hat A_2,\partial_1 \hat Z_\mu\bigr\rbrace_\star\;.
      \end{equation}
We show that $Z^\theta_1 = e^{\frac{i\sigma^1}{n-m\theta}}V^b$ and
$Z^\theta_2=e^{\frac{i\sigma^2}{n}}U^{-b}$ satisfy \eqref{endoeq} (and
of course the initial conditions $Z^\theta_\mu\vert_{\theta=0}=Z_\mu$).
      For $\mu=1$, we note that due to the dependence of $Z^\theta_1$
      and $\hat{A}_2$ only on $\sigma^1$, the star product on the
      right hand side of \eqref{endoeq} reduces to the usual matrix
      product and gives (since $\SWA_2$ is proportional to the unit matrix)
     \begin{equation}
        2\pi\,\SWA_2 \cdot \partial_1 Z^\theta_1 = \frac{im\sigma^1}{(n-m\theta)^2} \,e^{\frac{i\sigma^1}{n-m\theta}}\,V^b\;,
        \end{equation}
 which coincides with the left hand side  $\partial Z^\theta_\mu/\partial \theta$. The case $\mu=2$ is trivial due to the independence of $Z^\theta_2$ with respect to $\theta$ and its dependence solely on $\sigma^2$. 
    \end{proof}

In Proposition \ref{endolemma} we have just quantized the generators
$Z_\mu$ of sections in the adjoint representation. In general, the Seiberg--Witten
map quantizes an arbitrary section in the adjoint $\Psi\in
\textrm{End}({\cal E}_{n,m})$ to $\hat \Psi\in
\textrm{End}({\cal E}^\theta_{n,m})$. We can use the ambiguities in the Seiberg--Witten
maps to single out the one that quantizes $Z_\mu$ to
$\widehat{Z}_\mu=Z^\theta_\mu$ and that
is compatible with a specific ordering of the generators of the algebras
  $\textrm{End}({\cal E}_{n,m})$ and 
$\textrm{End}({\cal E^\theta}_{n,m})$. 

For example we study the
Seiberg--Witten map that to the ordered 
monomial ${Z_1}^{p} {Z_2}^{q}$ associates the ordered monomial 
$\widehat{\,{Z_1}^{p} {Z_2}^{q}\,}\!=\widehat{{Z_1}^{p}} \star \widehat{{Z_2}^{q}}={Z^\theta_1}^{\,p} \star{Z^\theta_2}^{\,q}$ (powers and $\star$-powers of $Z_1$
and $Z_2$ coincide). We compute 
\begin{align}
\frac{\partial}{\partial\theta} ({Z^\theta_1}^{\,p} \,\star  \: & {Z^\theta_2}^{\,q})-
\;\pi\,\bigl\lbrace \hat A_2,\partial_1  ({Z^\theta_1}^{\,p} \star
  {Z^\theta_2}^{\,q})\bigr\rbrace_\star\nonumber\\ &
=\pi i\,\partial_1 {Z^\theta_1}^{\,p}
  \star\partial_2 {Z^\theta_2}^{\,q}+ \frac{\partial {Z^\theta_1}^{\,p}
  }{\partial\theta}\star {Z^\theta_2}^{\,q}-\pi\big(\hat A_2\star  \partial_1{Z^\theta_1}^{\,p}
  \star {Z^\theta_2}^{\,q} +  \partial_1  {Z^\theta_1}^{\,p}
  \star {Z^\theta_2}^{\,q}\star \hat A_2\big) \nonumber\\ &
=\pi i\,\partial_1 {Z^\theta_1}^{\,p}
  \star\partial_2 {Z^\theta_2}^{\,q}+\frac{\partial {Z^\theta_1}^{\,p}
  }{\partial\theta}\star {Z^\theta_2}^{\,q}
-\pi\big(  \{\hat A_2\,  , \partial_1 {Z^\theta_1}^{\,p}\}_\star\star {Z^\theta_2}^{\,q}
- \partial_1 {Z^\theta_1}^{\,p}  \star  [\hat A_2\,,  {Z^\theta_2}^{\,q}]_\star \big) \nonumber\\ &
=\pi i\,\partial_1 {Z^\theta_1}^{\,p}\star D_2 {Z^\theta_2}^{\,q}
=\pi i\,\partial_1 ({Z^\theta_1}^{\,p}\star D_2 {Z^\theta_2}^{\,q})
=\pi i\,\partial_1 D_2({Z^\theta_1}^{\,p}\star {Z^\theta_2}^{\,q})
\nonumber\\ &
=\pi i\,D_1 D_2 ({Z^\theta_1}^{\,p}\star {Z^\theta_2}^{\,q})~,
\end{align}
where in the second line we used that $Z^\theta_2=Z_2$ is independent
from $\theta$ and $\sigma_1$, in the third line we
 used that \eqref{endoeq} 
implies 
${\partial {Z^\theta_1}^{\,p}}/{\partial \theta} =\;\pi\,\bigl\lbrace \hat A_2,\partial_1 { Z^\theta_\mu}^{\,p}\bigr\rbrace_\star$,
in the fourth line that $D_2 {Z^\theta_2}^{\,q}$ is independent from
$\sigma^1$ and that $D_2 {Z^\theta_1}^{\,p}=0$. We therefore see that
the relevant Seiberg--Witten map is given by the Seiberg--Witten differential
equation
\eq{\label{orderingSW}\frac{\partial\widehat{\Psi}}{\partial \theta}\,=\,\pi
\,\bigl\lbrace \hat A_2,\partial_1  \widehat\Psi\bigr\rbrace_\star+\pi i\,D_1 D_2 \widehat\Psi~;
}
comparison with \eqref{extendedSWPA} shows that it corresponds to the
choice $\hat C_{12}
=-\hat C_{21}
=-iD_1D_2\hat\Psi$. Here $\widehat\Psi$ is more generally any linear
combination of ordered monomials in $Z_1^\theta, Z^\theta_2$.

\begin{rem}
If we choose the opposite ordering,  ${Z_2}^{q} {Z_1}^{p}\to
\widehat{\,{Z_2}^{q} {Z_1}^{p}\,}\!=\widehat{{Z_2}^{q}} \star
\widehat{{Z_1}^{p}}={Z^\theta_2}^{\,q} \star{Z^\theta_1}^{\,p}$, the
corresponding Seiberg--Witten differential equation is
\eqref{orderingSW} with $- \pi i\,D_1 D_2\widehat\Psi$ instead of $+\pi i\,D_1 D_2 \widehat\Psi$.
Since ${Z^\theta_1}^{\,p} \,\star  \: 
{Z^\theta_2}^{\,q} -e^{2\pi i pq\check\theta} {Z^\theta_2}^{\,q} \,\star  \: 
{Z^\theta_1}^{\,p}=0$, cf. text after \eqref{Ztheta},  the ordering ${Z^\theta_1}^{\,p} \,\star  \: 
{Z^\theta_2}^{\,q} +e^{2\pi i pq\check\theta} {Z^\theta_2}^{\,q} \,\star  \: 
{Z^\theta_1}^{\,p}$ satisfies the  standard Seiberg--Witten
differential equation with $\hat C_{12}=0$; the corresponding map
quantizes ${Z_1}^{p}   \: 
{Z_2}^{q} \,+\,e^{2\pi i pq\frac{b}{n}} {Z_2}^{q}  \: {Z_1}^{p}$ to ${Z^\theta_1}^{\,p} \,\star  \: 
{Z^\theta_2}^{\,q} \,+\,e^{2\pi i pq\check\theta} {Z^\theta_2}^{\,q} \,\star  \: 
{Z^\theta_1}^{\,p}$.
 These three different Seiberg--Witten maps  are
easily seen to coincide on the generators $Z_\mu$.
\end{rem}

We now study the Seiberg-Witten quantization of sections in the fundamental
representation. In \cite{Ho:1998hqa}, the quantum section $\Hop$ in \eqref{qsec} was given. It is natural to ask how this solution fits in the framework of Seiberg--Witten quantization.
To this aim we study the $\theta$-dependence of the quantum section $\Hop$ in \eqref{qsec}.
\begin{lemma}
The quantum section $\Hop=(\Hop_k)^{}_{k=1,...n}$ in \eqref{qsec} satisfies the differential equation:
  \begin{equation}
    \label{Hodg}
    \frac{\partial}{\partial \theta} \Hop -\pi \SWA_2
    \star \partial_1 \Hop =\, 3\pi \,\SWF \star \Hop + i\pi\,D_1D_2\Hop\;.
  \end{equation}
  \end{lemma}
  \begin{proof}
For convenience, let us recall \eqref{qsec}: 
$      \Hop_k(\sigma^1,\sigma^2)=\,\underset{s\in\mathbb{Z}}{\sum}\overset{m}{\underset{j=1}{\sum}}\,E(\AAA,\BBB)\star \tilde \phi_j(\tfrac{n}{m} \AAA)\;,
$
where $A= \tfrac{m}{n}(\tfrac{\sigma^2}{2\pi} + k + ns) + j$,
$B=i\sigma^1$ and $E(A,B)$ is defined in \eqref{Efunction}.
The derivative of $\Hop=(\Hop_k)^{}_{k=1,...n}$ with respect to $\theta$ gives
    \begin{equation}
      \label{phider}
      \frac{\partial}{\partial \theta} \Hop =\, \underset{s,j}{\sum} \Bigl(\Bigl(\frac{\partial}{\partial \theta}\,E(\AAA,\BBB)\Bigr)\star \tilde \phi_j(\tfrac{m}{n}\AAA) + i\pi \partial_1 E(\AAA,\BBB) \star \partial_2 \tilde \phi_j(\tfrac{n}{m}\AAA)\Bigr)\;.
    \end{equation}
    The derivative of $E(\AAA,\BBB)$ with respect to $\theta$ is given
    by (use $\frac{p}{(p-1)!}=\frac{1}{(p-1)!}+\frac{p-1}{(p-1)!}$,
    and use Lemma \ref{uno} in the third line)
    \begin{align}
    \frac{\partial}{\partial \theta} E(\AAA,\BBB) =\;& \frac{\partial}{\partial \theta}\Bigl(\frac{1}{1-\tfrac{m}{n}\theta} \,\overset{\infty}{\underset{p=0}{\sum}} \frac{1}{p!}\,\AAA^p\star \BBB^p\Bigr) \nonumber\\
    =\;&\frac{m}{n-m\theta} \,E(\AAA,\BBB) +
         \frac{1}{2}\frac{m}{n}\Bigl(E(\AAA,\BBB) + \AAA\star
         E(\AAA,\BBB)\star \BBB\Bigr)\;,\nonumber\\
=\;&\frac{m}{n-m\theta} \,E(\AAA,\BBB) +
         \frac{1}{2}\frac{m}{n}E(\AAA,\BBB) + \frac{1}{2}\frac{m}{n-m\theta}\AAA\star
          \BBB\star E(\AAA,\BBB)\;.
    \end{align}
    The second term in \eqref{phider} gives, using Lemma
    \ref{inducelemma} and that $\partial_1=D_1$,
    \begin{equation}
      i\pi\partial_1 \,\underset{s,j}\sum\,E(\AAA,\BBB)\star \partial_2 \tilde \phi_j(\tfrac{n}{m}\AAA) =\;i\pi D_1D_2\,\Hop\;.
    \end{equation}
    Summarizing, the $\theta$-derivative of the section $\Hop$ is given by
    \begin{equation}
      \frac{\partial}{\partial \theta}\Hop =\;\Bigl(\frac{m}{n-m\theta} + \frac{1}{2}\frac{m}{n}\Bigr)\,\Hop + \underset{s,j}{\sum}\,\frac{1}{2}\frac{m}{n-m\theta}\,\AAA\star\BBB\star E(\AAA,\BBB)\star\tilde \phi_j(\tfrac{n}{m}\AAA) + i\pi\,D_1D_2\,\Hop\;.
    \end{equation}
   Furthermore, $\partial_1E(\AAA,\BBB)=iA\star E(\AAA,\BBB)$ implies 
    \begin{equation}
      \pi\,\SWA_2 \star \partial_1\Hop =\;\frac{1}{2}\frac{m}{n-m\theta}\,\underset{s,j}{\sum}\,\BBB\star \AAA  \star E(\AAA,\BBB)\star \tilde \phi_j\;.
    \end{equation}
    Hence, we arrive at
    \begin{align}
      \frac{\partial}{\partial \theta} \Hop -\pi \SWA_2 \star \partial_1 \Hop &=\,\Bigl(\frac{m}{n-m\theta} + \frac{1}{2}\frac{m}{n}\Bigr)\,\Hop + i\pi\,D_1D_2\,\Hop \nonumber\\
      &~~~~~+ \underset{s,j}{\sum}\,\frac{1}{2}\frac{m}{n-m\theta}\Bigl(\AAA\star \BBB - \BBB\star \AAA\Bigr)\star E(\AAA,\BBB)\star \tilde \phi_j(\tfrac{n}{m}\AAA)\nonumber\\
      &=\,\frac{3}{2}\frac{m}{n-m\theta}\,\Hop + i\pi\,D_1D_2\,\Hop\;\nonumber\\[.2em]
      &=\,3\pi\,\SWF \star \Hop + i\pi\,D_1D_2\Hop\;
    \end{align}
    where in the last line we used $\hat F= i[D_1,D_2] = \tfrac{1}{2\pi} \tfrac{m}{n-m\theta}\mathbf{1}$.
  \end{proof}

Comparison of \eqref{Hodg} with the (generalized) Seiberg--Witten differential
equation for fields in the fundamental representation
shows that \eqref{Hodg} equals \eqref{extendedSWP} with 
\eq{\label{CHo}\hat C_{12}(\hat \phi, \hat A)
=-\hat C_{21}(\hat \phi, \hat A)
=-3\hat F\star\Hop-iD_1D_2\Hop~.}
The same operator $-3\hat F-iD_1D_2$ acting in the adjoint reads 
$\hat C_{12}(\hat \Psi, \hat A)=-\hat C_{21}(\hat \Psi, \hat A)=-iD_1D_2\hat\Psi$ (use $[\hat
F,\widehat\Psi]_\star=0$ since $\hat F$ is constant and proportional
to the identity). In particular, for the sections 
$\hat \Omega_\al=\Omega_\al\in \textrm{End}(E^\theta) =(\mathbb{R}_\theta^2)^{\oplus (n\times n)}$,
we have $\hat C_{12}(\hat \Omega_\al, \hat A)=0$.
These  expressions transform covariantly under
gauge transformations (they are sections of $E^\theta=(\mathbb{R}^2_\theta)^{\oplus n}$;
cf. also \eqref{CDD}, and of $\textrm{End}(E^\theta)
=(\mathbb{R}_\theta^2)^{\oplus (n\times n)}$) and the
corresponding (generalized) Seiberg--Witten map is a quantization of bundles on tori as shown in Theorem \ref{TSWText}.  We thus conclude that 

\begin{theorem}\label{endolemma2}
The Seiberg--Witten map with $\hat C_{\mu\nu}$ given by the operator
$-3\hat F-iD_1D_2$ quantizes, following Theorem  \ref{TSWText},   $({\mathcal E}_{n,m},A_\mu)$
to $(\hat {\mathcal E}_{n,m},\hat A_\mu)=({\mathcal E}_{n,m}^\theta,
A^\theta_\mu)$; the sections $\phi\in{\cal E}_{n,m}$ defined in \eqref{classec} 
to the sections  $\hat\phi=\Hop\in{\cal E}^\theta_{n,m}$ as
defined in \eqref{qsec}; the adjoint sections ${Z_1}^{\,p} {Z_2}^{\,q}\in  \textrm{End}({\cal
  E}_{n,m})$ to ${Z^\theta_1}^{\,p}
\star{Z^\theta_2}^{\,q}\in\textrm{End}({\cal E^\theta}_{n,m})$.
    \end{theorem}

We have recovered within the framework of Seiberg--Witten map Ho's solutions 
$\phi^\theta\in{\cal E}^\theta_{n,m}$ and $Z_\mu^\theta\in \textrm{End}({\cal E}^\theta_{n,m})$ to the noncommutative periodicity conditions.
 It follows that this Seiberg--Witten framework, initially developed in the context of deformation
quantization with $\theta$ a formal deformation parameter, can be
specialized to $\theta\in \mathbb{R}-\{\frac{m}{n}\}$. 
Indeed the $\star$-product can be completed to a nonformal product
\`a la Rieffel, and the connection ${\hat A}_\mu$, and the sections
$\Hop$, $Z_\mu^\theta$ are well defined for  $\theta\in \mathbb{R}-\{\frac{m}{n}\}$.

\section{Morita equivalence, T-duality and Seiberg--Witten
    map}\label{MESWM}

Here we briefly review how Morita equivalence implements T-duality
of Yang-Mills theories and show the compatibility of the Seiberg--Witten
maps with T-duality transformations.

\sk
Two (associative and unital) algebras $A$ and $\tilde A$ are Morita equivalent if their
categories of right modules $\M^{}_A$ and $\M_{\tilde A}$ are equivalent. By a theorem of Morita
two algebras $ A$ and $\tilde A$ are Morita equivalent if and only if there exists
a finitely generated and projective $A$-module ${\cal  E}\in
\M^{}_A$ such that every other $A$ module is a quotient of ${\cal
  E}^{\oplus N}$ for some integer $N$. In this case ${\cal  E}\in
{}^{}_{^{}A}\M^{\phantom{j_j}}_{\tilde A}$ is called a Morita
equivalence bimodule. 
The equivalence between the categories of representations $\M_A$ and
$\M_{\tilde A}$ is easily constructed via 
${\cal E}\in{}^{}_{^{}A}\M^{\phantom{j_j}}_{\tilde A}$. 
The main point is that given a module $E\in  \M_A$ the tensor product over $A$ with ${\cal
  E}\in{}^{}_{^{}A}\M^{\phantom{j_j}}_{\tilde A}$ gives a module
$\tilde E=E\otimes_A {\cal E}\in  \M_{\tilde A}$, and  morphisms of
$\M_A$  modules $\varphi: E\to F$ are mapped to morphisms of
$\M_{\tilde A}$  modules $\varphi\otimes_A id: E\otimes_A {\cal E}\to
F\otimes_A {\cal E}$.  
The bimodules ${\cal{E}}_{n,m}^\theta$, $n>0$,
$m\not=0$, $n,m$ relatively prime, in Section \ref{UNNC}
(Heisenberg modules) are
examples of Morita equivalence bimodules and
prove Morita equivalence of 2-dimensional tori 
$T_{(-\tilde\theta)}$ and $T_{(-\theta)}$ related by a fractional $SL(2,\mathbb{Z})$ transformation: $(-\tilde\theta)
=\frac{a(-\theta)+b}{m(-\theta)+n}$ with  $\big({}^a_m\, {}^b_n\big)
\in SL(2,\mathbb{Z})$. Since $T_{(-\theta)}$ is isomorphic to
$T_{\theta}$ we further have that
 2-dimensional tori $T_{\tilde\theta}$ and $T_{(-\theta)}$ related by a fractional
$GL(2,\mathbb{Z})$ transformation are Morita equivalent. This is also
a necessary condition:
The 2-dimensional tori $T_{(-\tilde\theta)}$ and $T_{(-\theta)}$ are
Morita equivalent, if and only if$\,$\footnote{There are different notions of Morita equivalence: the one just
recalled for algebras (and more generally rings),  a stronger notion
for $C^*$-algebras, and, in the case of (multidimensional) tori, an
even stronger one called complete Morita equivalence of smooth
noncommutative tori \cite{Schwarz:1998qj} (called gauge Morita equivalence in
\cite{Konechny:2000dp}). These notions for 2-dimensional tori are all equivalent
(the bimodules ${\cal{E}}_{n,m}^\theta$ can be completed to full
Hilbert modules providing $C^*$-algebra Morita equivalence, and they are Heisenberg modules with constant curvature connections that provide complete Morita equivalence).
} 
$-\tilde \theta =\frac{a(-\theta) + b }{m(-\theta) + n}\;,$ with 
$\big({}^a_m\, {}^{b}_{n}\big) \in  GL(2,\mathbb{Z})$, \cite{rieffel1981}. 
\sk
In order to study the equivalence of categories between  Morita
equivalent torus algebras we consider the Heisenberg module
${\cal{E}}_{0,1}^\frac{1}{s-\theta}$, $s\in \mathbb{Z}$,
that is a bimodule in
${}^{}_{T_{(-\theta)}}\M^{}_{T_{(\frac{1}{\theta-s})}}$, accounting
for the transformation $-\theta\to
\tfrac{1}{\theta-s}$.
We recall from Section \ref{UNNC}, that
${\cal{E}}_{0,1}^\frac{1}{s-\theta}={\cal{E}}_{1,1}^{\frac{1}{s-\theta}+1}$
is the module of sections of the $U(1)$-bundle over ${T_{(\frac{1}{\theta-s}+1)}} $ with charge
$m=1$, algebraically it is the vector space of Schwartz
functions on $\mathbb{R}$, with
right {$T_{(\frac{1}{\theta-s})}$}-action given by
    \eq{\label{phiUx}
      (\tilde\phi \triangleleft U_1)(x) =&\;
      {\tilde\phi\Bigl(x+\frac{1}{s-\theta}\Bigr)}\;,\qquad
      (\tilde\phi \triangleleft U_2)(x)=\; \tilde\phi(x)\,e^{2\pi i
        x}\;.
}
The algebra $T_{(-\theta)}$ of endomorphisms is generated by 
\eq{
      (Z_1 \triangleright \tilde\phi)(x) =&\;\tilde\phi(x-1)\;, \qquad
      (Z_2 \triangleright \tilde\phi)(x)=\;{\tilde\phi(x) \,e^{2\pi i
          x(\theta-s)}}=
\;{\tilde\phi(x) \,e^{2\pi i x\theta}}\;,
    }
while the connection is
\eq{
  D_{1} \tilde \phi\/(x) =\; \frac{-ix}{\theta}\,\tilde \phi_j(x)\;,\label{D1s}~~~~~
    D_{2\,} \tilde \phi\/(x)=\;
                                           \frac{1}{2\pi}\,\frac{\partial}{\partial
                                           x}\tilde \phi(x)\;.
                                           }
Since any projective module over
$T_{(-\theta)}$ is equivalent to a direct sum of Heisenberg modules
${\cal{E}}_{n,m}^\theta$ ($n\geq 0$,
$m\not=0$, $n,m$ relatively prime) 
 or the trivial module $T_{(-\theta)}$,  it is sufficient to
describe the transformations of these modules under a Morita
equivalence bimodule ${\cal{E}}_{0,1}^\frac{1}{s-\theta}$, $s\in
\mathbb{Z}$ in order to know it on every module in  $\M_{T_{(-\theta)}}$.

We have an isomorphism
(see for example the outlined explicit derivation in
\cite[\S 3]{Schwarz:1998qj})
\eq{{\cal{E}}_{n,m}^{\theta}\otimes^{}_{T_{(-\theta)}}{\cal{E}}_{0,1}^{\frac{1}{s-\theta}}\:\simeq\:
      {\cal{E}}_{m,-n+ms}^{\frac{1}{s-\theta}}\label{tensorA}~,
}
in particular, both left hand side and right hand side are left
$T_{\check{\theta}\,}$-modules with $\check{\theta}={\frac{a(-\theta)+b}{m(-\theta)+n}}$. 
We see that under the Morita equivalence bimodule $
{\cal{E}}_{0,1}^{\frac{1}{s-\theta}} $ the module
${\cal{E}}_{n,m}^{\theta}$ is mapped (up to equivalence) to the module
${\cal{E}}_{\tilde n,\tilde m}^{\tilde\theta}={\cal{E}}_{m,-n+ms}^{\frac{1}{s-\theta}}$,
thus
$\theta\mapsto\tilde\theta=\big({}^{~0}_{-1}\,{}^{1}_{s}\big)\theta=\tfrac{1}{s-\theta}~,~~\big({}^{n}_{m}\big)\mapsto
\big({}^{\tilde n}_{\tilde m}\big)=
\big({}^{~0}_{-1}\, {}^{1}_{s}\big) \big({}^{n}_{m}\big)=\big({}^{~~~m}_{-n+ms}\big)\,.
$
\sk
Similarly, let's define the bimodule ${\cal{E}}_{1,0}^{\theta+1}\in
{}^{}_{T_{(-\theta)}}\M^{}_{T_{(-\theta-1)}}$ to be $T_{(-\theta)}$
 as a left $T_{(-\theta)}$-module, with right $T_{(-\theta-1)}$-action
defined by  $e\triangleleft
U_\mu^{T_{(-\theta-1)}}=eU_\mu$, where $e\in
{\cal{E}}_{1,0}^{\theta+1}$, $U_\mu^{T_{(-\theta-1)}}$ are
the generators of $T_{(-\theta-1)}$,  $U_\mu$ those of
$T_{(-\theta)}$ and $eU_\mu$ is the product in $T_{(-\theta)}$.  Then it is easy to show that 
\eq{
{\cal{E}}_{n,m}^{\theta}\otimes^{}_{T_{(-\theta)}}{\cal{E}}_{1,0}^{{\theta+1}}\,\simeq\,
      {\cal{E}}_{n+m, m}^{\theta+1}\label{tensor+1}~~.
}
Indeed using the nonvanishing global section $1\in
{\cal{E}}_{1,0}^{\theta+1}$ we write a generic section of
${\cal{E}}_{n,m}^{\theta}\otimes^{}_{T_{(-\theta)}}{\cal{E}}_{1,0}^{{\theta+1}}$
as $\tilde\phi \otimes^{}_{T_{(-\theta)}} 1$.  We prove
the equivalence \eqref{tensor+1}  by showing that $\tilde\phi
\otimes^{}_{T_{(-\theta)}} 1$ transforms as a
section of ${\cal{E}}_{n+m, m}^{\theta+1}$, 
\begin{eqnarray}
\big((\tilde\phi \otimes^{}_{T_{(-\theta)}} 1)\triangleleft
U_1^{T_{(-\theta-1)}}\big)^{}_j(x)&\!\!=\!\!&(\tilde\phi \otimes^{}_{T_{(-\theta)}}
1)^{}_{j-1}\big(x-\frac{n+m}{m}+(\theta+1)\big)\label{proofa}\;,\\[.2em]
\big((\tilde\phi \otimes^{}_{T_{(-\theta)}} 1)\triangleleft
U_2^{T_{(-\theta-1)}}\big)^{}_j(x)&\!\!=\!\!&(\tilde\phi \otimes^{}_{T_{(-\theta)}}
1)^{}_{j}\big(x)\,e^{2\pi i (x-j(n+m)/m)}\;~. \nonumber
\end{eqnarray} 
E.g. \eqref{proofa} follows from $(\tilde\phi \otimes^{}_{T_{(-\theta)}} 1)\triangleleft
U_1^{T_{(-\theta-1)}}=\tilde\phi \otimes^{}_{T_{(-\theta)}} 1\triangleleft
U_1^{T_{(-\theta-1)}}=
\tilde\phi \otimes^{}_{T_{(-\theta)}} 
U_1=
\tilde\phi\triangleleft U_1 \otimes^{}_{T_{(-\theta)}} 
1$, and $\big(\tilde\phi\triangleleft U_1 \otimes^{}_{T_{(-\theta)}} 
1)^{}_j(x)=\tilde\phi_{j-1}(x-\frac{n}{m}+\theta) \otimes^{}_{T_{(-\theta)}} 
1=(\tilde\phi \otimes^{}_{T_{(-\theta)}}
1)^{}_{j-1}\big(x-\frac{n+m}{m}+(\theta+1)\big)$.
We have seen that tensoring with the Morita equivalence bimodule
${\cal{E}}_{1,0}^{\theta+1}\in
{}^{}_{T_{(-\theta)}}\M^{}_{T_{(-\theta-1)}}$  gives the map
${\cal{E}}_{n,m}^{\theta}\mapsto {\cal{E}}_{\tilde n,\tilde m}^{\tilde\theta}= {\cal{E}}_{n+m, m}^{\theta+1}$, thus 
$\theta\mapsto\tilde\theta=\big({}^{1}_{0}\,{}^{1}_{1}\big)\theta=\theta+1,~\big({}^{n}_{m}\big)\mapsto
\big({}^{\tilde n}_{\tilde m}\big)=
\big({}^{1}_{0}\, {}^{1}_{1}\big) \big({}^{n}_{m}\big)=\big({}^{n+m}_{~~\!m}\big)\,.
$
\sk

We similarly constuct a 
bimodule ${\cal E}^{-\theta}_{1,0}$ that realizes the isomorphism
$T_{(-\theta)}\simeq T_{\theta}$ as a Morita equivalence.
By definition ${\cal E}^{-\theta}_{1,0}\in
{}^{}_{T_{(-\theta)}}\M^{}_{T_{(\theta)}}$ is $T_{(-\theta)}$ itself as
a left $T_{(-\theta)}$-module, while the right $T_{\theta}$-action on
${\cal{E}}_{1,0}^{-\theta}$ is defined by $e\triangleleft
U_1^{T_{\theta}}=eU^{-1}_1$, $e \triangleleft U_2^{T_{\theta}}=eU_2$, where $e\in
{\cal{E}}_{1,0}^{\theta+1}$, $U_\mu^{T_{\theta}}$ are
the generators of $T_{\theta}$, $U_\mu$ those of
$T_{(-\theta)}$ and $eU^{\pm 1}_\mu$ is the product in $T_{(-\theta)}$.  
 A generic section of
${\cal{E}}_{n,m}^{\theta}\otimes^{}_{T_{(-\theta)}}{\cal{E}}_{1,0}^{{-\theta}}$
is $\tilde\phi\otimes_{T_{(-\theta)}} 1$ with $\tilde\phi$
a section of ${\cal{E}}_{n,m}^{\theta}$.
 We show that 
\eq{
{\cal{E}}_{n,m}^{\theta}\otimes^{}_{T_{(-\theta)}}{\cal{E}}_{1,0}^{{-\theta}}\,\simeq\,
      {\cal{E}}_{n, -m}^{-\theta}\label{tensor-}~~
}
by showing that the map  
$\tilde\phi\otimes 1 \mapsto \iota(\tilde\phi\otimes
1)$ defined by $\iota(\tilde\phi\otimes
1)_j(x) =\tilde \phi_{|m|-j}(x)$, (we can assume
$|m|-j=1,2,...|m|$ due to $\mathbb{Z}_{|m|}$ cyclicity)
is a right $T_{\theta}$-module isomorphism.
Proof:  We have to
show that $\iota((\tilde\phi\otimes 1)\triangleleft
U^{T_\theta}_\mu)=\iota(\tilde\phi\otimes 1)\triangleleft
U^{T_\theta}_\mu$.
Indeed,\begin{eqnarray}
\iota((\tilde\phi\otimes 1)\triangleleft
U_1^{T_\theta})_j(x)&\!\!\!=\!\!\!&\iota(\tilde\phi\triangleleft
U^{-1}_1\otimes
1)_j(x)
=(\tilde\phi\triangleleft U_1^{-1})_{|m|-j}(x)
=\tilde\phi_{|m|-j+1}(x+\frac{n}{m}-\theta)
\nonumber\\
&\!\!\!=\!\!\!&\iota(\tilde\phi\otimes 1)_{j-1}(x+\frac{n}{m}-\theta)=
(\iota(\tilde\phi\otimes 1)\triangleleft 
U^{T_\theta}_1)_j(x)~,
\end{eqnarray}
and similarly for $U^{T_\theta}_2$. 
We have seen that tensoring with the Morita equivalence bimodule
${\cal{E}}_{1,0}^{-\theta}$ gives the transformation (up to equivalence)
${\cal{E}}_{n,m}^{\theta}\mapsto {\cal{E}}_{\tilde n,\tilde m}^{\tilde\theta}= {\cal{E}}_{n+m, m}^{\theta+1}$, thus 
$\theta\mapsto\tilde\theta=\big({}^{1}_{0}\,{}^{\;0}_{-1}\big)\theta=-\theta~,~\big({}^{n}_{m}\big)\mapsto
\big({}^{\tilde n}_{\tilde m}\big)=
\big({}^{1}_{0}\, {}^{\;0}_{-1}\big) \big({}^{n}_{m}\big)=\big({}^{\;n}_{-m}\big)\,.
$
\sk

Since $SL(2,\mathbb{Z})$ is generated by $\big({}^{~0}_{-1}\,
{}^{1}_{0}\big)$ and $\big({}^1_0\,{}^{1}_{1}\big)$, and
$GL(2,\mathbb{Z})$ by considering also 
 $\big({}^1_0\,{}^{\;0}_{-1}\big)$,
  we see that
\eqref{tensorA}, \eqref{tensor+1} and \eqref{tensor-}
generate the whole nongeometric $GL(2,\mathbb{Z})$ duality group that
acts on the modules ${\cal E}_{n,m}^\theta\in \M_{T_{(-\theta)}}$
as
\eq{\theta\mapsto\tilde\theta= \frac{a\theta + b }{m\theta +
        n}~~,~~~
 \begin{pmatrix}
        n\\
        m
      \end{pmatrix} \mapsto  \begin{pmatrix}
        \tilde n\\
        \tilde m
      \end{pmatrix}=\begin{pmatrix}
        a & b\\
        c & d
      \end{pmatrix} \begin{pmatrix}
        n\\
        m
      \end{pmatrix}~~,~~~\begin{pmatrix}
        a & b\\
        c & d
      \end{pmatrix} \in GL(2,\mathbb{Z})\;.
\label{GL2Z} 
}
\sk

We now describe how the equivalence of the categories of modules
${\cal M}_{T_{(-\theta)}}$ and ${\cal M}_{T_{(-\tilde\theta)}}$ (with
$\tilde\theta$ in the same $GL(2,\mathbb{Z})$ orbit of $\theta$) 
is extended to an equivalence between modules with connections,
called  complete (or gauge) Morita equivalence.
This is due to the canonical connections of
the bimodules ${\cal{E}}_{0,1}^{\frac{1}{s-\theta}}$ and
${\cal{E}}_{1,0}^{\theta+1}$, that correspondingly define these bimodules as
complete (or gauge) Morita equivalence bimodules, see definition after \eqref{CME}.
It is this gauge Morita equivalence that implements $T$-duality
transformations between gauge theories on noncommutative bundles
${\cal E}_{n,m}^\theta$ on ${T_{(-\theta)}}$ and ${\cal E}_{\tilde
  n,\tilde m}^{\tilde \theta}$ on ${T_{(-\tilde\theta)}}$.

To any (right) connection $\nabla_\mu$ on  ${\cal E}_{n,m}^\theta$
there canonically corresponds a connection on the tensor product bundle
$
{\cal{E}}_{n,m}^{\theta}\otimes^{}_{T_{(-\theta)}}{\cal{E}}_{0,1}^{\frac{1}{s-\theta}}
$
given by the sum $(\theta-s)\nabla_\mu\otimes \id +\id\otimes D_\mu$ 
where $D_\mu$ is the canonical constant curvature connection of the
Heisenberg module ${\cal{E}}_{0,1}^{\frac{1}{s-\theta}} $. Explicitly, for
$\tilde\phi\in {\cal{E}}_{n,m}^{\theta}$, $\tilde\phi'\in
{\cal{E}}_{0,1}^{\frac{1}{s-\theta}}$,
$\big((\theta-s)\nabla_\mu\otimes \id +\id\otimes D_\mu\big)(\tilde\phi
\otimes^{}_{T_{(-\theta)}}
\tilde\phi')=(\theta-s)\nabla_\mu\tilde\phi
\otimes^{}_{T_{(-\theta)}}\tilde\phi'\,+\,\tilde\phi \otimes^{}_{T_{(-\theta)}}D_\mu\tilde\phi'$.
The rescaling by $\theta-s$ is needed in order for 
the canonical derivations $\partial_{\sigma^\mu}$ of $T_{(-\theta)}$
(entering the Leibniz rule for the connection $\nabla_\mu$)
to match the derivations $\hat\delta_\mu$ on
${\rm{End}}({\cal{E}}_{0,1}^{\frac{1}{s-\theta}})\simeq T_{(-\theta)}
$ induced from the connection $D_\mu$ as in \eqref{CME}.
Indeed this matching insures that the sum $(\theta-s)\nabla_\mu\otimes
\id +\id\otimes D_\mu$ is well defined on
${\cal{E}}_{n,m}^{\theta}\otimes^{}_{T_{(-\theta)}}{\cal{E}}_{0,1}^{\frac{1}{s-\theta}}$,
meaning that since the tensor product is over $T_{(-\theta)}$, the
  connection on  $\tilde\phi
  a\otimes_{T_{(-\theta})}\tilde\phi'$ equals that on $\tilde\phi
  \otimes_{T_{(-\theta})}a \tilde\phi'$, for all $a\in T_{(-\theta)}$.

As a special case we can choose $\nabla_\mu=D_\mu$, where this latter is the
canonical constant curvature connection of the Heisenberg module
${\cal E}_{n,m}^\theta$. Then
${\cal{E}}_{n,m}^{\theta}\otimes^{}_{T_{(-\theta)}}{\cal{E}}_{0,1}^{\frac{1}{s-\theta}}$
has constant curvature connection
${\mathcal{D}}_\mu:=(\theta-s)D_\mu\otimes \id +\id\otimes D_\mu$,
indeed the curvature is easily computed to be ${\mathcal
  F}_{12}=i[{\mathcal{D}}_1,{\mathcal{D}}_2]=\frac{1}{2\pi}
(\theta-s)\frac{n-ms}{n-m\theta}\mathbf{1}$. Furthermore,
this curvature coincides with that of the Heisenberg module ${\cal{E}}_{m,-n+ms}^{\frac{1}{s-\theta}}$.
This proves that the equivalence  $
{\cal{E}}_{n,m}^{\theta}\otimes^{}_{T_{(-\theta)}}{\cal{E}}_{0,1}^{\frac{1}{s-\theta}}\simeq
{\cal{E}}_{m,-n+ms}^{\frac{1}{s-\theta}}$ extends to an equivalence between
Heisenberg modules, i.e., modules with constant curvature connections. 

Similarly, to any  connection $\nabla_\mu$ on  ${\cal E}_{n,m}^\theta$
there canonically corresponds a connection on the tensor product bundle
$
{\cal{E}}_{n,m}^{\theta}\otimes^{}_{T_{(-\theta)}}{\cal{E}}_{1,0}^{\theta+1}
$
given by the sum $\nabla_\mu\otimes \id +\id\otimes \partial_{\sigma^\mu}$
where $\partial_{\sigma^\mu}$ is the canonical flat connection of the
trivial line bundle ${\cal{E}}_{1,0}^{\theta+1}\simeq
T_{(-\theta)}$. We see that the curvature is unchanged. It is also
intructive to consider the case
${\cal{E}}_{n,m}^{\theta}\otimes^{}_{T_{(-\theta)}}{\cal{E}}_{1,0}^{-\theta}$. 
The canonical flat connection of the right $T_\theta$-module
${\cal{E}}_{1,0}^{-\theta}$ is
$\partial_{-\sigma^1}, \partial_{\sigma^2}$, indeed these are the
canonical derivations  of $T_\theta$ with generators $U_1^{T_\theta}=U_1^{-1}$,
$U_2^{T_\theta}=U_2$, if $\partial_{\sigma^1}, \partial_{\sigma^2}$
are those of $T_{(-\theta)}$ with generators $U_1,U_2$. Matching the derivations
$\hat\delta_\mu$ induced as in \eqref{inducedderiv} by the connection
$\partial_{-\sigma^1},\partial_{\sigma^2}$ of
${\cal{E}}_{1,0}^{-\theta}$ with the derivations
$\partial_{\sigma^1},\partial_{\sigma^2}$ of $T_{(-\theta)}$ (defining
the Leibnitz rule of the connection $\nabla_\mu$ of  ${\cal{E}}_{n,m}^{\theta}$) we
obtain the connection $-\nabla_1\otimes \id
+\id\otimes \partial_{-\sigma^1}$,  $\nabla_2\otimes \id
+\id\otimes \partial_{\sigma^2}$ on the tensor product bundle
${\cal{E}}_{n,m}^{\theta}\otimes^{}_{T_{(-\theta)}}{\cal{E}}_{1,0}^{-\theta}$.  
Notice that its curvature is the opposite of that of ${\cal{E}}_{n,m}^{\theta}$.
\sk
As we have seen in the previous section, the Seiberg--Witten map
quantizes the Heisenberg modules ${\cal{E}}_{n,m}$. It is
therefore natural to ask if Seiberg--Witten map quantization is
compatible with complete Morita equivalence, i.e. if to a given Seiberg--Witten quantization there corresponds a Seiberg--Witten quantization of the T-dual modules. This is expected since the Seiberg--Witten quantization of ${\cal{E}}_{n,m}$ is $\hat{\cal{E}}_{n,m}={\cal{E}}^\theta_{n,m}$.
We indeed give a positive answer in the following section.

  \subsection{Compatibility of Seiberg--Witten maps with T-duality}
Since $GL(2,\mathbb{Z})$ is generated by 
$\big({}^{~0}_{-1}\, {}^{1}_{0}\big)$, $\big({}^1_0\, {}^1_1\big)$ and $\big({}^1_0\, {}^{~0}_{-1}\big)$,
corresponding to $\theta\mapsto  -1/\theta$, $\theta\mapsto \theta+1$, $\theta\mapsto -\theta$, it
sufficies to prove compatibility with these transformations.
For example, the nontrivial dualities $\theta\to\tfrac{-1}{\theta-s}$,
$s\in \mathbb{Z}$ (including that of the generator $\big({}^{~0}_{-1}\,
{}^{1}_{0}\big)$) are compatible with the Seiberg--Witten
quantization maps according to the following commutative diagram
  \begin{equation}
    \begin{matrix}
 \big({\mathcal{E}}_{n,m}^\theta\in\prescript{}{\textrm{End}({\mathcal E}^\theta_{n,m})}{\M}^{}_{\TT2_{(-\theta)}}, A^\theta_\mu\big)
 & \xrightarrow{\;\otimes_{T_{(-\theta)}\,}
   (\CE_{0,1}^{\frac{1}{s-\theta}}, \,D_\mu)\;~} &
 \big({\mathcal{E}}_{m,-n+ms}^{~\tilde \theta = \frac{1}{s-\theta}}\in\prescript{}{\textrm{End}({\mathcal
     E}^{\tilde\theta}_{m,-n+ms})}{\M}^{}_{\TT2_{(-\tilde\theta)}}, {A}^{\tilde\theta}_\mu\big)
\\[1.4em]
    \!\!\! SW^{\theta'}_\theta\Bigg\downarrow~~~~ & &  \!\!\!SW^{{{\tilde\theta\,}}_{}'}_{\tilde\theta}  \Bigg\downarrow   \\
     \big({\mathcal{E}}_{n,m}^{\theta'}\in\prescript{}{\textrm{End}({\mathcal
         E}^{\theta'}_{n,m})}{\M}^{}_{\TT2_{(-\theta')}},
     A^{\theta'}_\mu\big) & \xrightarrow{~\;\otimes_{T_{(-\theta')}\,}
       (\CE_{0,1}^{\frac{1}{s-\theta'}},\,D_\mu)\;~} & 
\big({\mathcal{E}}_{m,-n+ms}^{\,{\tilde{\theta}\,}'=_{} \tilde{\theta'\;} = \frac{1}{s-\theta'}}\in\prescript{}{\textrm{End}({\mathcal E}^{{\tilde{\theta}\,'}}_{m,-n+ms})}{\M}^{}_{\TT2_{(-{\tilde\theta}\,')}}, {A}^{\tilde{\theta}\,'}_\mu\big)\label{CDC}
    \end{matrix}
  \end{equation}
where the vertical arrow $SW^{\theta'}_\theta$ denotes one of the Seiberg--Witten maps
of Theorem \ref{TSWText},  from
$\theta$ to $\theta'$, cf. \eqref{SWTT'}, and similarly for
$SW^{{{\tilde\theta\,}}_{}'}_{\tilde\theta}$, while the horizontal arrow
$\otimes^{}_{T_{(-\theta)}} (\CE_{0,1}^{\frac{1}{s-\theta}}, \,D_\mu)$
 denotes the map of Heisenberg
modules with connections to Heisenberg modules with connection
obtained via the complete Morita equivalence bimodule
$(\CE_{0,1}^{\frac{1}{s-\theta}}, \,D_\mu)$, and similarly for the
lower horizontal arrow.

The  commutativity of the diagram is due to the definition
    ${{\tilde\theta}^{}\,}'=\tilde{\theta'\;}$ in the Seiberg--Witten
    map  $SW_{\tilde\theta}^{{{{\tilde\theta}\,}}_{}'}$.
The Seiberg--Witten map also relates the
      equivalence bimodules with connection 
$(\CE_{0,1}^{\frac{1}{s-\theta}},\,D_\mu) $ and $(\CE_{0,1}^{\frac{1}{s-\theta'}},\,D_\mu)$
 used to obtain the $T$-dual modules in the right hand side of the diagram. Just recall that 
$(\CE_{0,1}^{\theta},\,D_\mu)=(\CE_{1,1}^{\theta+1}, \,D_\mu)$
is the module of sections of the $U(1)$-bundle over ${T_{(\frac{1}{\theta-s}+1)}} $ with charge
$m=1$, cf. \eqref{phiUx}-\eqref{D1s},
and conclude that the corresponding Seiberg--Witten map is 
\begin{equation}\label{EDSW1}
    \begin{matrix}
 (\CE_{0,1}^{\frac{1}{s-\theta}},
\,D_\mu)=(\CE_{1,1}^{\frac{1}{s-\theta}+1}, \,D_\mu)
 &
 \xrightarrow{\;\,SW_{{\frac{1}{s-\theta}+1}}^{{\frac{1}{s-\theta'}+1}}\;~}
 &
(\CE_{1,1}^{\frac{1}{s-\theta'}+1},\,D_\mu)=
(\CE_{0,1}^{\frac{1}{s-\theta'}},\,D_\mu)~.
\end{matrix}
\end{equation}

The commutativity of the diagrams like \eqref{CDC} but for the other two $T$-duality
transformations $\big({}^1_0\, {}^1_1\big)$ and
$\big({}^1_0\, {}^{~0}_{-1}\big)$, is straighforward, just 
consider the upper horizontal arrows with
$({\cal{E}}_{1,0}^{{\theta+1}},D_\mu)$, respectively
$({\cal{E}}_{1,0}^{{-\theta}},D_\mu)$, and similarly with $\theta\to
\theta'$ for the lower horizontal arrows. Then
consider
$SW_{\theta+1}^{\theta'+1}$, respectively  $SW_{-\theta}^{-\theta'}$ in the right hand
  side of the corresponding diagrams.
Correspondingly, the Seiberg--Witten map \eqref{EDSW1} is replaced by
$SW_{\theta+1}^{\theta'+1}\!:({\cal{E}}_{1,0}^{{\theta+1}},D_\mu)\to ({\cal{E}}_{1,0}^{{\theta'+1}},D_\mu) $, respectively $SW_{-\theta}^{-\theta'}\!: ({\cal{E}}_{1,0}^{{-\theta}},D_\mu)\to ({\cal{E}}_{1,0}^{{-\theta'}},D_\mu)\,.$

In conclusion we have that Seiberg--Witten maps are compatible
with complete Morita equivalence.

\sk
Finally we mention that the Seiberg--Witten map
$SW_\theta^{\theta+1}$ always differs from a Morita equivalence,
however if we consider only trivial bundles ($m=0$), then it corresponds to
tensoring with   $\CE_{1,0}^{\theta+1}$, (c.f. \eqref{tensor+1}).

\subsubsection*{Acknowledgements}
The authors are members of COST Action MP1405 \textit{QSpace - 
Quantum Structure of Spacetime} and of  INFN, CSN4, Iniziativa
Specifica GSS, that have partially supported this project. This
research has also a financial support of the Universit\`a
del Piemonte Orientale.  A.D. is grateful to Heriot-Watt University
for hospitality. P.A. and A.D. acknowledge hospitality form LMU, Munich. 
P.A. is affiliated to INdAM, GNFM (Istituto Nazionale di Alta Matematica, Gruppo Nazionale di Fisica Matematica).

\bibliographystyle{utphys}
\bibliography{SWNCT}

\end{document}